\documentclass[11pt,letterpaper]{article}
\usepackage[margin=1in]{geometry}
\usepackage{amsmath,amssymb,amsthm}
\usepackage{booktabs}
\usepackage{array}
\usepackage{xcolor}
\usepackage{graphicx}
\usepackage{tikz}
\usetikzlibrary{arrows.meta}
\usepackage{algorithm}
\usepackage[noend]{algpseudocode}
\usepackage{enumitem}
\usepackage{microtype}
\usepackage[hidelinks]{hyperref}

\theoremstyle{plain}
\newtheorem{theorem}{Theorem}
\newtheorem{lemma}[theorem]{Lemma}

\theoremstyle{definition}
\newtheorem{definition}[theorem]{Definition}

\algblockdefx[Upon]{Upon}{EndUpon}[1]{\textbf{upon} #1 \textbf{do}}{}
\algnotext{EndUpon}
\algrenewcommand\algorithmicindent{1.2em}

\newcommand{\PROPOSAL}{\textnormal{\textsc{proposal}}}
\newcommand{\PRECOMMIT}{\textnormal{\textsc{precommit}}}
\newcommand{\nil}{\mathsf{nil}}
\newcommand{\validRound}{\mathit{validRound}}
\newcommand{\validValueId}{\mathit{validValueId}}
\newcommand{\validPair}[1]{(\validRound_{#1},\validValueId_{#1})}
\newcommand{\round}{\mathit{round}}
\newcommand{\step}{\mathit{step}}
\newcommand{\id}{\mathit{id}}
\newcommand{\valid}{\mathit{valid}}
\newcommand{\decision}{\mathit{decision}}
\newcommand{\committedTip}{\mathit{committedTip}}
\newcommand{\propose}{\mathit{propose}}
\newcommand{\pcstep}{\mathit{precommit}}
\newcommand{\proposer}{\mathit{proposer}}
\newcommand{\tprop}{\mathit{timeoutPropose}}
\newcommand{\twait}{\mathit{timeoutWait}}
\newcommand{\tprec}{\mathit{timeoutPrecommit}}
\newcommand{\GST}{\mathrm{GST}}
\newcommand{\new}[1]{\textcolor{blue}{#1}}

\title{Revisiting Fast Tendermint:\\An Agreement Violation and Two Repairs}
\author{Alejandro Ranchal-Pedrosa\\Sei Labs}
\date{October 2026}

\begin{document}
\maketitle

\begin{abstract}
Fast Tendermint reduces Tendermint's good-case latency to two communication steps for $n\geq5f+1$. We show that its published round-change rule violates Agreement: one Byzantine process suffices for every $f\geq1$, and for $f\geq2$ two Byzantine processes suffice without double signing. The failure occurs when a process advances past a round without recording the value that may have been decided there. The Quint specification that accompanies the protocol does not allow this kind of round skipping and is not affected; the violation lies in the published pseudocode. We present two repairs. Walk requires a process to leave each round on a quorum of that round's precommits and to precommit before leaving. It keeps the recorded valid state and removes the wait after a precommit quorum. Justify carries the evidence needed to accept a proposal in the proposal itself, removing both the recorded valid state and the proposer's synchronization wait. Both satisfy Agreement and Validity under asynchrony and Termination under the partially synchronous gossip model of Tendermint and Fast Tendermint. Both retain proposal timeouts. The main text presents the executions that violate Agreement and explains the two repairs. The appendices provide proofs, a brief comparison with other round-change designs, and an extension to chains.
\end{abstract}

\section{Introduction}

Tendermint~\cite{tendermint} tolerates $f<n/3$ Byzantine processes and has a good-case latency of three communication steps: a proposal, prevotes and precommits. Fast Tendermint~\cite{ft} adapts it to $f<n/5$ and reduces this to two. The prevote step is dropped, a proposal is answered directly with a precommit, and a value is decided once $n-f$ precommits support it.

Tendermint records both a locked value and a valid value, each with its round. A process locks a valid proposal after receiving $2f+1$ matching prevotes and then broadcasts a non-nil precommit. The lock persists across round changes; voting for a conflicting value requires a sufficiently recent justification. Fast Tendermint replaces these two pairs with a single pair, updated when the process receives $2f+1$ matching precommits. Sending a precommit alone does not update it.

The published protocol lets a process advance on a quorum of any round at or above its own, but its observation rule records support only from the current round. A process can therefore advance past a round without recording the value that may have been decided there, and later help certify a different value. Section~\ref{sec:violation} gives two concrete executions: one Byzantine process suffices for every $f\geq1$; for $f\geq2$, two suffice without sending conflicting precommits.

We focus on two constructions. \emph{Walk} is a direct adaptation for a Tendermint codebase: it preserves the valid state, advances through rounds sequentially on a quorum of the current round, and requires a precommit before leaving. \emph{Justify} carries the required evidence in the proposal, so acceptance no longer depends on the recipient's private valid state. This removes the proposer's evidence wait at the cost of larger proposals. Both preserve a single precommit per process per round and keep proposal timeouts for silent proposers.

The main text presents the violation and the two repairs. Appendix~\ref{app:ft} reproduces the published algorithm. Appendix~\ref{app:evidence} establishes the evidence lemmas, and Appendices~\ref{app:walk} and~\ref{app:justify} prove Walk and Justify. Appendix~\ref{app:fastexit} gives their optional support-certificate and dead-certificate exits. The remaining appendices discuss alternative designs, the role of the precommit timer, an optional shortcut, chains and implementation details. The three alternative designs are sketches, separate from the correctness claims for the main constructions.

\paragraph{Status.}
We reported the violation and two fixes to the authors of~\cite{ft} in October 2026. They reproduced it by modeling the published pseudocode in Quint. According to the authors, the Quint specification that accompanies~\cite{ft}, which they consider the reference specification of the protocol, does not allow round skipping and is not affected by the executions of Section~\ref{sec:violation}. The violation therefore lies in the published pseudocode, Algorithm~1 of~\cite{ft}, which is the protocol analyzed here.

\section{Model and the published protocol}
\label{sec:model}

There are $n\geq5f+1$ processes, at most $f$ of which are Byzantine. The others are correct and follow the protocol. Messages are signed. A \emph{precommit} is a message $\langle\PRECOMMIT,s,x\rangle$, where $s$ is a round and $x$ is a value identifier $\id(v)$ or $\nil$. The symbol $\nil$ is not a value identifier. A process \emph{equivocates} in round $s$ if it sends precommits with different contents for that round. Correct processes send at most one precommit per round.

Proposers rotate in a fixed round-robin order, so infinitely many rounds have correct proposers. Values satisfy a deterministic per-height application predicate $\valid$, common to all correct processes. The function $\mathit{getValue}()$ returns an application-valid value, and $\id$ is collision resistant.

Agreement means that no two correct processes decide different values. Validity means that decided values satisfy $\valid$. Termination means that every correct process eventually decides.

\paragraph{Timing and dissemination.}
We use the partially synchronous model of Tendermint and Fast Tendermint~\cite{tendermint,ft}, including their gossip property: if a correct process sends or receives a protocol message at time $t$, every correct process receives it by $\max\{t,\GST\}+\Delta$. Before the Global Stabilization Time $\GST$, the adversary may delay messages arbitrarily. The bound $\Delta$ covers end-to-end dissemination, including relaying and recovery of retained messages, rather than one network hop.

Agreement and Validity hold under asynchrony and use no timing assumption. Every Termination result assumes the gossip property and the timeout bounds stated with that result. Appendix~\ref{app:implementation} explains the corresponding requirements for relaying and value retrieval.

Executions have finitely many events in every bounded time interval. Enabled local rules are processed without delay; each rule is atomic except where it explicitly waits, and an aborted proposer continuation is permanently cancelled. Correct processes retain their protocol state. In the per-height protocols they stop taking consensus actions after deciding, while continuing to disseminate the decision evidence and value. The chained extension explicitly removes this stopping condition.

For the repaired algorithms, an ``upon'' condition is a persistent predicate over retained authenticated messages and verified attachments. It is re-evaluated after a message arrives, a round begins or a wait begins. Certificates count distinct senders; conflicting Byzantine signatures remain available as evidence. Broadcast includes immediate self-delivery.

Every consensus handler requires $\decision_p=\nil$. A decision sets this variable once and cancels pending consensus timers and proposer continuations, while dissemination continues. A round change permanently cancels the departed round's proposer continuation. Immediately before a proposal is broadcast, the process checks $\round_p=r$, $\step_p=\propose$ and $\decision_p=\nil$. These checks, payload selection and broadcast form one atomic action. These execution conventions apply to the repaired algorithms, not to the historical listing in Appendix~\ref{app:ft}.

\paragraph{Terminology and notation.}
We follow Tendermint's terminology: a consensus instance is a height, each height proceeds in rounds, and each round has a proposer and protocol steps. For one height we suppress the height argument in messages, $\proposer$ and timeout handlers, and write $\decision_p$ for Tendermint's $\decision_p[h_p]$. We use $\valid(v)$ for application validity.

The recorded state consists of $\validRound_p$ and $\validValueId_p$, adapting Tendermint's $\mathit{validValue}$ to the identifier-only state of Fast Tendermint. Recording a certificate does not require possession of the full value. An assignment such as $\validPair{p}\gets(s,\id(v))$ updates both components atomically. We retain \emph{fresh proposal}, \emph{reproposal}, $\mathit{WaitForValid}$ and the observation rule from Fast Tendermint. We use Minimmit's \emph{view}, \emph{notarization} and \emph{nullification} only when discussing that protocol.

\paragraph{Published protocol.}
We restate the published protocol for one height by naming its rules; line numbers refer to Algorithm~1 of~\cite{ft}, which Appendix~\ref{app:ft} reproduces in our notation as Algorithm~\ref{alg:ft} with the original line numbers. The state is $\round_p$, $\step_p \in \{\propose, \pcstep\}$, $\decision_p$, $\validRound_p$ and $\validValueId_p$. The last two are initially $-1$ and $\nil$, respectively, and together represent the published pair $\mathit{valid}_p=(\round,\id)$.
\begin{itemize}[nosep]
\item \emph{Start.} $\mathit{StartRound}(r)$ sets $\round_p \gets r$ and $\step_p \gets \propose$. The proposer of a round $r>0$ first runs $\mathit{WaitForValid}$ for $\tprec(r)$: while $\validRound_p < r-1$ it adopts any certificate (defined below) of a round above $\validRound_p$. It then broadcasts $\langle\PROPOSAL, r, v, -1\rangle$ with a fresh value $v$ if $\validRound_p = -1$, and otherwise reproposes $\langle\PROPOSAL, r, \validValueId_p, \validRound_p\rangle$. Other processes schedule $\tprop(r)$, on whose expiry a process still in step $\propose$ precommits $\nil$.
\item \emph{Fresh proposal} (line 18). On $\langle\PROPOSAL, \round_p, v, -1\rangle$ from the proposer while in step $\propose$: precommit $\id(v)$ if $\valid(v)$ and ($\validRound_p = -1$ or $\validValueId_p = \id(v)$), else precommit $\nil$; enter step $\pcstep$.
\item \emph{Reproposal} (line 24). On $\langle\PROPOSAL, \round_p, \id(v), vr\rangle$ with $0 \leq vr < \round_p$ and $2f+1$ precommits for $\id(v)$ in round $vr$, while in step $\propose$: if $\validRound_p \leq vr$ or $\validValueId_p = \id(v)$, then set $\validPair{p} \gets (vr,\id(v))$ when $\validRound_p \leq vr$ and precommit $\id(v)$; else precommit $\nil$. Enter step $\pcstep$.
\item \emph{Observation} (line 32). On $2f+1$ precommits for $\id(v)$ in round $\round_p$ with $\round_p > \validRound_p$: $\validPair{p} \gets (\round_p, \id(v))$.
\item \emph{Quorum} (line 34). On $n-f$ precommits of any round $r \geq \round_p$, for the first time: schedule $\mathit{OnTimeoutPrecommit}(r)$ after $\tprec(r)$, which starts round $r+1$ if $r \geq \round_p$ and no decision has been made.
\item \emph{Decision} (line 36). On $\langle\PROPOSAL, r, v, -1\rangle$ from $\proposer(r)$ and $n-f$ precommits for $\id(v)$ in some round $r'$: decide $v$.
\end{itemize}
Precommitting changes no state other than $\step_p$. The observation rule is the only rule that raises $\validRound_p$ to the current round, and it sees only precommits of the current round. The quorum rule is the only path that moves a process past a round, and it accepts quorums of rounds above the current one. The mismatch between these rules allows a process to skip a round without recording its certified value.

\paragraph{Timers.}

Both main constructions retain $\tprop$, the proposal timeout needed when a proposer is silent. Walk also retains the proposer's evidence wait $\twait$; Justify removes that wait. Both remove $\tprec$, the delay before advancing after a precommit quorum. The timeout conditions are stated with the termination proofs. Appendix~\ref{app:alternatives} discusses alternatives that retain a timer for round advancement.

\section{Agreement violation}
\label{sec:violation}

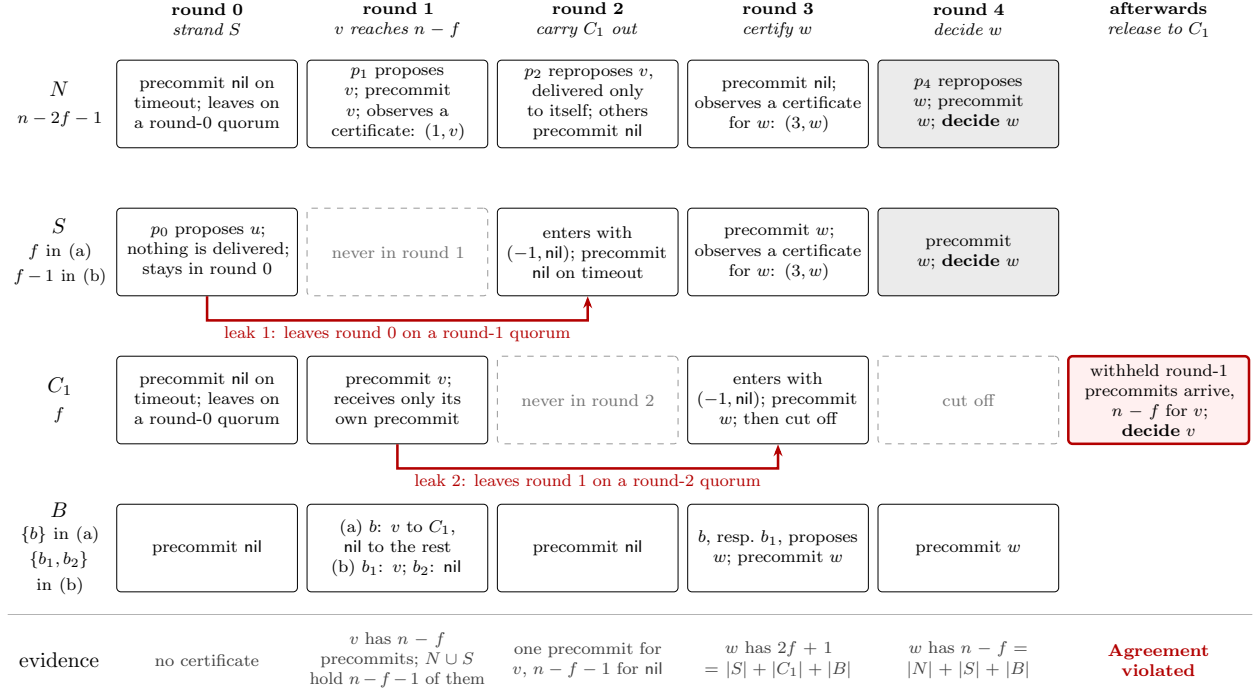
\begin{figure}[tbp]
\centering
\resizebox{\textwidth}{!}{%
\begin{tikzpicture}[
  font=\scriptsize,
  cell/.style={draw, rounded corners=2pt, text width=2.85cm, minimum height=1.45cm, align=center, inner sep=2pt,
    execute at begin node={\hyphenpenalty=10000\exhyphenpenalty=10000}},
  gone/.style={cell, dashed, draw=black!40, text=black!55},
  done/.style={cell, fill=black!8},
  bad/.style={cell, very thick, draw=red!70!black, fill=red!6},
  hdr/.style={text width=2.9cm, align=center},
  lane/.style={text width=1.5cm, align=center, font=\small},
  leak/.style={-{Stealth[length=2mm]}, very thick, red!70!black},
  leaktext/.style={font=\scriptsize, text=red!70!black, inner sep=2pt},
  ev/.style={text width=2.9cm, align=center, text=black!75,
    execute at begin node={\hyphenpenalty=10000\exhyphenpenalty=10000}}
]
\node[hdr] at (0,1.4) {\textbf{round 0}\\ \textit{strand $S$}};
\node[hdr] at (3.15,1.4) {\textbf{round 1}\\ \textit{$v$ reaches $n-f$}};
\node[hdr] at (6.3,1.4) {\textbf{round 2}\\ \textit{carry $C_1$ out}};
\node[hdr] at (9.45,1.4) {\textbf{round 3}\\ \textit{certify $w$}};
\node[hdr] at (12.6,1.4) {\textbf{round 4}\\ \textit{decide $w$}};
\node[hdr] at (15.75,1.4) {\textbf{afterwards}\\ \textit{release to $C_1$}};
\node[lane] at (-2.45,0) {$N$\\[1pt] \scriptsize $n-2f-1$};
\node[lane] at (-2.45,-2.45) {$S$\\[1pt] \scriptsize $f$ in (a)\\ $f-1$ in (b)};
\node[lane] at (-2.45,-4.9) {$C_1$\\[1pt] \scriptsize $f$};
\node[lane] at (-2.45,-7.35) {$B$\\[1pt] \scriptsize $\{b\}$ in (a)\\ $\{b_1,b_2\}$ in (b)};
\node[cell] (n0) at (0,0) {precommit $\nil$ on timeout; leaves on a round-0 quorum};
\node[cell] (n1) at (3.15,0) {$p_1$ proposes $v$; precommit $v$; observes a certificate: $(1,v)$};
\node[cell] (n2) at (6.3,0) {$p_2$ reproposes $v$, delivered only to itself; others precommit $\nil$};
\node[cell] (n3) at (9.45,0) {precommit $\nil$; observes a certificate for $w$: $(3,w)$};
\node[done] (n4) at (12.6,0) {$p_4$ reproposes $w$; precommit $w$; \textbf{decide $w$}};
\node[cell] (s0) at (0,-2.45) {$p_0$ proposes $u$; nothing is delivered; stays in round 0};
\node[gone] (s1) at (3.15,-2.45) {never in round 1};
\node[cell] (s2) at (6.3,-2.45) {enters with\\ $(-1,\nil)$; precommit $\nil$ on timeout};
\node[cell] (s3) at (9.45,-2.45) {precommit $w$; observes a certificate for $w$: $(3,w)$};
\node[done] (s4) at (12.6,-2.45) {precommit $w$; \textbf{decide $w$}};
\node[cell] (c0) at (0,-4.9) {precommit $\nil$ on timeout; leaves on a round-0 quorum};
\node[cell] (c1) at (3.15,-4.9) {precommit $v$; receives only its own precommit};
\node[gone] (c2) at (6.3,-4.9) {never in round 2};
\node[cell] (c3) at (9.45,-4.9) {enters with\\ $(-1,\nil)$; precommit $w$; then cut off};
\node[gone] (c4) at (12.6,-4.9) {cut off};
\node[bad] (c5) at (15.75,-4.9) {withheld round-1 precommits arrive, $n-f$ for $v$;\\ \textbf{decide $v$}};
\node[cell] (b0) at (0,-7.35) {precommit $\nil$};
\node[cell] (b1) at (3.15,-7.35) {(a) $b$: $v$ to $C_1$, $\nil$ to the rest\\ (b) $b_1$: $v$; $b_2$: $\nil$};
\node[cell] (b2) at (6.3,-7.35) {precommit $\nil$};
\node[cell] (b3) at (9.45,-7.35) {$b$, resp.\ $b_1$, proposes $w$; precommit $w$};
\node[cell] (b4) at (12.6,-7.35) {precommit $w$};
\draw[leak] (s0.south) -- ++(0,-0.4) -| node[pos=0.25, below, leaktext] {leak 1: leaves round 0 on a round-1 quorum} (s2.south);
\draw[leak] (c1.south) -- ++(0,-0.4) -| node[pos=0.25, below, leaktext] {leak 2: leaves round 1 on a round-2 quorum} (c3.south);
\draw[black!30] (-3.3,-8.45) -- (17.3,-8.45);
\node[lane] at (-2.45,-9.2) {evidence};
\node[ev] at (0,-9.2) {no certificate};
\node[ev] at (3.15,-9.2) {$v$ has $n-f$ precommits; $N\cup S$ hold $n-f-1$ of them};
\node[ev] at (6.3,-9.2) {one precommit for $v$, $n-f-1$ for $\nil$};
\node[ev] at (9.45,-9.2) {$w$ has $2f+1$ $=|S|+|C_1|+|B|$};
\node[ev] at (12.6,-9.2) {$w$ has $n-f$ $=|N|+|S|+|B|$};
\node[ev, text=red!70!black] at (15.75,-9.2) {\textbf{Agreement violated}};
\end{tikzpicture}%
}
\caption{The executions of Theorem~\ref{thm:violation}, all before $\GST$. Rows are the sets of the proof and columns are rounds. $N$ contains the proposers $p_1$, $p_2$ and $p_4$ of rounds 1, 2 and 4, $S$ contains the proposer $p_0$ of round 0, and the proposer of round 3 is Byzantine. A pair $(r,x)$ is the recorded state $\validPair{}$, and dashed cells are rounds a set never enters. The red arrows are the two ways of leaving round~1 without recording $(1,v)$: $S$ leaves round~0 on a round-1 quorum, and $C_1$ leaves round~1 on a round-2 quorum. Execution~(b) differs from (a) only in the sizes of $S$ and $B$, in the Byzantine precommits of round~1 shown in row $B$, and in withholding the round-1 precommit of one fixed member of $C_1$ from $N\cup S$.}
\label{fig:violation}
\end{figure}

\begin{theorem}
\label{thm:violation}
Let $n \geq 5f+1$. Algorithm~1 of~\cite{ft} violates Agreement.
\begin{enumerate}[nosep,label=(\alph*)]
\item For every $f \geq 1$, one Byzantine process suffices, and it equivocates exactly once.
\item For every $f \geq 2$, two Byzantine processes suffice, and neither of them ever equivocates.
\end{enumerate}
\end{theorem}

\begin{proof}
Figure~\ref{fig:violation} summarizes both executions. All line references in this proof are to Algorithm~1 of~\cite{ft}, reproduced with its original numbering as Algorithm~\ref{alg:ft} in Appendix~\ref{app:ft}. They describe the actions of correct processes; Byzantine messages are chosen by the adversary.

Both executions use the same arrangement of processes. Partition the correct processes into three sets: $S$, the processes that will be stranded in round~0; $C_1$, the processes that will precommit the decided value but leave round~1 on a later quorum; and $N$, the rest. Let $B$ be the Byzantine processes. For (a), $|S| = f$, $|C_1| = f$ and $B = \{b\}$. For (b), $|S| = f-1$, $|C_1| = f$ and $B = \{b_1, b_2\}$. In both cases $|N| \geq 3f$.

The proposers are chosen as follows: $\proposer(0) \in S$, $\proposer(1), \proposer(2), \proposer(4) \in N$, and $\proposer(3) = b$ in (a), $b_1$ in (b); a round-robin order allows this.

Everything happens before $\GST$, so the adversary chooses every delay, and a broadcast reaches its sender immediately. Every withheld message is delayed in all copies, including gossip relays and attachments. Round-0 quorum deliveries are staggered so that the round-1 proposer enters early and the other intended voters enter shortly before its $\mathit{WaitForValid}$ expires; no inequality between $\tprop(1)$ and $\tprec(1)$ is needed. Later quorum deliveries align the intended participants' entries into rounds 2, 3 and 4 before the corresponding proposal deliveries.

The rounds serve the following purposes: round~0 strands $S$, round~1 produces $n-f$ precommits for a value $v$, round~2 carries $C_1$ out of round~1 without round-1 information, round~3 certifies a value $w \neq v$, and round~4 decides $w$.

\emph{Round 0.} The proposer broadcasts a fresh value $u$ (lines~\ref{ft:11} to~\ref{ft:15}). The adversary delivers it to no process but the proposer, which precommits $\id(u)$ (lines~\ref{ft:18} to~\ref{ft:20}); every other correct process precommits $\nil$ when $\tprop(0)$ expires (lines~\ref{ft:42} to~\ref{ft:45}, scheduled at line~\ref{ft:17}). The Byzantine processes also send $\nil$. The round-0 precommits of $N \cup C_1 \cup B$ are delivered to $N \cup C_1$; there are at least $n-f$ of them, so the quorum rule schedules the precommit timer (lines~\ref{ft:34} and~\ref{ft:35}) and $N \cup C_1$ start round~1 when it expires (lines~\ref{ft:46} to~\ref{ft:48}). Nothing is delivered to $S$, which stays in round~0. No value has $2f+1$ precommits in round~0.

\emph{Round 1.} The proposer $p_1 \in N$ finds no certificate in $\mathit{WaitForValid}$ (lines~\ref{ft:38} to~\ref{ft:41}) and, when the wait expires, broadcasts a fresh value $v$ (lines~\ref{ft:11} to~\ref{ft:15}). It is delivered to every process in $N \cup C_1$ before its $\tprop(1)$ expires, and they all precommit $\id(v)$ (lines~\ref{ft:18} to~\ref{ft:20}). In (a), $b$ sends $\langle\PRECOMMIT,1,\id(v)\rangle$ to $C_1$ and $\langle\PRECOMMIT,1,\nil\rangle$ to $N \cup S$: this is its one equivocation. In (b), $b_1$ sends $\langle\PRECOMMIT,1,\id(v)\rangle$ to everyone and $b_2$ sends $\langle\PRECOMMIT,1,\nil\rangle$ to everyone. 

Each process in $C_1$ receives no round-1 precommit but its own. Each process in $N \cup S$ receives, in (a), the precommits for $\id(v)$ of $N \cup C_1$ and the $\nil$ of $b$; in (b), the precommits for $\id(v)$ of $N \cup C_1$ minus one fixed member of $C_1$, the precommit of $b_1$, and the $\nil$ of $b_2$. In both cases these are precommits from $n-f$ distinct processes, exactly $n-f-1$ of them for $\id(v)$. 

In $N$ the observation rule fires (lines~\ref{ft:32} and~\ref{ft:33}), since $n-f-1 \geq 2f+1$, and sets $\validPair{} = (1,\id(v))$; the quorum rule schedules the timer (lines~\ref{ft:34} and~\ref{ft:35}); the decision rule (line~\ref{ft:36}) does not fire, and $N$ start round~2 when the timer expires (lines~\ref{ft:46} to~\ref{ft:48}). In $S$, still in round~0, the quorum rule fires because $1 \geq 0$ (lines~\ref{ft:34} and~\ref{ft:35}). The observation condition at line~\ref{ft:32} does not hold because the precommits are of round~1 and not of round~0. When the timer expires, lines~\ref{ft:46} to~\ref{ft:48} let $S$ start round~2 with $\validPair{} = (-1,\nil)$. Thus $S$ have skipped round~1 without recording its certified value. Processes in $C_1$ stay in round~1. Nevertheless, $v$ already has $n-f$ precommits in round~1: those of $N \cup C_1$ and, in (a), the copy $b$ sent to $C_1$; in (b), the one of $b_1$.

\emph{Round 2.} The proposer $p_2 \in N$ has $\validRound = 1$, so the loop condition at line~\ref{ft:39} is already false. $\mathit{WaitForValid}$ returns at once and $p_2$ reproposes $\langle\PROPOSAL,2,\id(v),1\rangle$ (lines~\ref{ft:14} and~\ref{ft:15}). By self-delivery $p_2$ applies the reproposal rule (lines~\ref{ft:24} to~\ref{ft:31}), which holds since it has a certificate for $\id(v)$ in round~1 and $\validRound = 1 \leq 1$, and precommits $\id(v)$. The adversary delivers the reproposal to no other process before its $\tprop(2)$ expires, so the rest of $N$ and all of $S$ precommit $\nil$ on timeout (lines~\ref{ft:42} to~\ref{ft:45}). All of $B$ also precommit $\nil$. These round-2 precommits, one for $\id(v)$ and $n-f-1$ for $\nil$, are delivered to every process. In $N \cup S$ the quorum rule fires (lines~\ref{ft:34} and~\ref{ft:35}), but the observation condition at line~\ref{ft:32} does not hold. They start round~3 on the timer (lines~\ref{ft:46} to~\ref{ft:48}). In $C_1$, in round~1, the quorum rule fires with $r = 2 \geq 1$ (lines~\ref{ft:34} and~\ref{ft:35}), although line~\ref{ft:32} cannot observe the round-2 evidence from round~1. They start round~3 through lines~\ref{ft:46} to~\ref{ft:48} with $\validPair{} = (-1,\nil)$, having never held $2f+1$ precommits for $\id(v)$. Thus $C_1$ have also left round~1 without recording its certified value. After round~2, $N$ hold $(1,\id(v))$ and $S \cup C_1$ hold $(-1,\nil)$, all in round~3.

\emph{Round 3.} The Byzantine proposer broadcasts $\langle\PROPOSAL,3,w,-1\rangle$ to everyone, with $w \neq v$ and $\valid(w)$. Processes in $S \cup C_1$ have $\validRound = -1$ and precommit $\id(w)$ (lines~\ref{ft:19} and~\ref{ft:20}); processes in $N$ fail the check at line~\ref{ft:19} and precommit $\nil$ at line~\ref{ft:22}; the Byzantine processes precommit $\id(w)$. So $\id(w)$ has $|S|+|C_1|+|B| = 2f+1$ precommits in round~3. All round-3 precommits are delivered to $N \cup S$. The observation rule fires at each of them (lines~\ref{ft:32} and~\ref{ft:33}), since $3$ exceeds their $\validRound$, and sets $\validPair{} = (3,\id(w))$; the quorum rule schedules the timer (lines~\ref{ft:34} and~\ref{ft:35}), and its expiry starts round~4 (lines~\ref{ft:46} to~\ref{ft:48}). From now on nothing is delivered to $C_1$.

\emph{Round 4.} The proposer $p_4 \in N$ reproposes $\langle\PROPOSAL,4,\id(w),3\rangle$ (lines~\ref{ft:14} and~\ref{ft:15}; line~\ref{ft:39} again skips the wait). It reaches $N \cup S$ in time; each applies the reproposal rule (lines~\ref{ft:24} to~\ref{ft:31}), which holds since they have the round-3 certificate for $\id(w)$ and $\validRound = 3 \leq 3$, and precommits $\id(w)$; the Byzantine processes precommit $\id(w)$. That is $|N|+|S|+|B| = n-f$ precommits for $\id(w)$ in round~4. Together with the round-3 fresh proposal of $w$ from $\proposer(3)$, the decision rule fires (lines~\ref{ft:36} and~\ref{ft:37}, with $r=3$ and $r'=4$ in line~\ref{ft:36}) and every process in $N \cup S$ decides $w$.

Finally the adversary delivers to $C_1$ the round-1 precommits for $\id(v)$ it withheld: those of $N \cup C_1$ and the Byzantine one addressed to $C_1$, which are $n-f$ in total, together with $p_1$'s proposal of $v$. The decision rule at line~\ref{ft:36} is not restricted to the current round: it fires with $r=r'=1$, and line~\ref{ft:37} makes every process in $C_1$ decide $v$. Two correct processes have decided different values. All remaining withheld copies can be delivered afterwards; choosing $\GST$ after this finite execution and delivering them by $\GST+\Delta$ respects the dissemination assumption.

The stated round-1 through round-4 quorum and certificate counts are exact for every $n\geq5f+1$, with $|N|=n-2f-1$; the round-0 set in (b) has one extra precommit. In (b), $b_1$ precommits $\id(v)$, $\nil$, $\id(w)$, $\id(w)$ in rounds $1$ to $4$ and $b_2$ precommits $\nil$, $\nil$, $\id(w)$, $\id(w)$; each sends one precommit per round. Proposing a fresh $w$ while holding a certificate for $v$ is not detectable from messages.
\end{proof}

The execution exposes two ways to leave round~1 without recording $(1,\id(v))$. Processes in $S$ skip it from round~0. Their round-1 quorum contains at least $n-3f\geq2f+1$ precommits for $v$, but the observation rule is still looking at round~0. Processes in $C_1$ precommit $v$ in round~1, then leave on a round-2 quorum before receiving its round-1 certificate. That later quorum supplies no evidence about round~1. In both cases the process advances without the valid state needed to preserve Agreement.

\section{Evidence for advancing rounds}
\label{sec:evidence}

A \emph{precommit set} of round $s$ is a set of precommits of round $s$ with at most one precommit per sender. For a precommit set $T$ and a value $x$, let $a_x(T)$ be the number of senders whose precommit in $T$ is for $\id(x)$. Byzantine processes may send several precommits for one round; a precommit set keeps one of them, and the counts below are about the precommit set a process holds, except where they explicitly concern all precommits ever sent.

\begin{definition}
\label{def:certs}
A \emph{support certificate}, or simply a \emph{certificate}, for $x$ in round $s$ is a precommit set of $2f+1$ precommits $\langle\PRECOMMIT,s,\id(x)\rangle$. A \emph{nil certificate} for round $s$ is a precommit set of $2f+1$ precommits $\langle\PRECOMMIT,s,\nil\rangle$.

A \emph{dead certificate} for round $s$ is a precommit set $T$ of round $s$ with $|T| - a_x(T) \geq 2f+1$ for every value $x$: for every $x$, at least $2f+1$ senders in $T$ precommit something other than $\id(x)$.

A \emph{quorum of precommits}, or simply a \emph{quorum}, of round $s$ is a precommit set of $n-f$ precommits of round $s$, possibly for different values or $\nil$. A \emph{decision certificate} for $x$ in round $s$ is a quorum whose precommits are all for $\id(x)$.
\end{definition}

A support certificate does not by itself decide a value. A decision certificate permits a decision once the value is known. A nil certificate is a dead certificate. A quorum in which no value has $2f+1$ precommits is also a dead certificate when $n \geq 5f+1$: for every $x$, at least $n-f-2f \geq 2f+1$ of its senders precommit something else. Adding a precommit from a new sender preserves the defining condition of a dead certificate.

A dead certificate rules out a decision in its round, but does not rule out a support certificate in that round. Unlike a nullification in Minimmit~\cite{minimmit}, it consists of ordinary precommits and introduces no separate signed message.

\paragraph{Certificates and decisions.}
A support certificate for $x$ rules out a decision for a different value in the same round: the support and decision signer sets share a correct process, which cannot precommit twice. A dead certificate rules out every decision in its round. Neither certificate, on its own, establishes what happened in an earlier round. This limitation is central to the violation: the published rule permits advancement using evidence from a later round. Lemmas~\ref{lem:evidence} and~\ref{lem:exit} in Appendix~\ref{app:evidence} give the counting arguments.

Both repairs require evidence for every round a process passes. A process leaves only its current round, using a quorum of that round. Walk records certified support in local state. Justify carries the relevant certificates in its proposals. Neither allows a process to skip an intermediate round merely because a later round has accumulated a mixed quorum. Appendix~\ref{app:fastexit} separately shows how a support certificate or a dead certificate of the current round can replace the quorum as exit evidence.

\section{Walk}
\label{sec:walk}

Walk is the protocol we recommend for an existing Tendermint codebase. Its main changes concern leaving a round: a process leaves only on a quorum of its current round, precommits before leaving, and advances without the post-quorum wait. Algorithm~\ref{alg:walk} gives the full protocol, with these changes highlighted in blue. The accompanying text explains the remaining specification details.

\setcounter{algorithm}{1}
\begin{algorithm}[htp]
\caption{Walk, one height, process $p$. Blue lines highlight the round-change repair. ``precommit $y$'' stands for broadcast $\langle\PRECOMMIT,\round_p,y\rangle$. Rule~\ref{ln:obs} is applied before Rule~\ref{ln:exit} whenever both are enabled, including immediately after round entry. Rule~\ref{ln:exit} applies at any time, including during $\mathit{WaitForValid}$, which it aborts. Messages of rounds above $\round_p$ are kept and take effect when $\round_p$ reaches their round.}
\label{alg:walk}
\small
\begin{algorithmic}[1]
\State $\round_p := 0$; $\step_p := \propose$; $\decision_p := \nil$; $\validPair{p} := (-1,\nil)$
\Upon{start}
  \State $\mathit{StartRound}(0)$
\EndUpon
\Function{StartRound}{$r$}
  \State $\round_p \gets r$; $\step_p \gets \propose$
  \If{$\proposer(r) = p$}
    \State \textbf{if} $r > 0$ \textbf{then} $\mathit{WaitForValid}(r)$
    \State \textbf{if} $\round_p\neq r\vee\step_p\neq\propose\vee\decision_p\neq\nil$ \textbf{then return}
    \State $\mathit{proposal} \gets \validValueId_p$ \textbf{if} $\validValueId_p \neq \nil$ \textbf{else} $\mathit{getValue}()$
    \State broadcast $\langle\PROPOSAL, r, \mathit{proposal}, \validRound_p\rangle$
  \Else
    \State schedule $\mathit{OnTimeoutPropose}(r)$ after $\tprop(r)$
  \EndIf
\EndFunction
\Function{WaitForValid}{$r$}
  \State \new{wait until $\validRound_p = r-1$ or $\twait(r)$ has elapsed; meanwhile:}
  \State \new{\textbf{upon} a certificate for $x$ in a round $s$ with $\validRound_p < s \leq r-1$ \textbf{do} $\validPair{p} \gets (s,\id(x))$} \label{ln:wfv}
\EndFunction
\Upon{$\langle\PROPOSAL, \round_p, v, -1\rangle$ from $\proposer(\round_p)$ while $\step_p = \propose$} \label{ln:fresh}
  \State \textbf{if} $\valid(v) \wedge (\validRound_p = -1 \vee \validValueId_p = \id(v))$ \textbf{then} precommit $\id(v)$ \textbf{else} precommit $\nil$
  \State $\step_p \gets \pcstep$
\EndUpon
\Upon{$\langle\PROPOSAL, \round_p, \id(v), vr\rangle$ from $\proposer(\round_p)$ with $0 \leq vr < \round_p$ and a certificate for $v$ in round $vr$, while $\step_p = \propose$} \label{ln:repro}
  \If{$\validRound_p \leq vr \vee \validValueId_p = \id(v)$}
    \State \textbf{if} $\validRound_p \leq vr$ \textbf{then} $\validPair{p} \gets (vr,\id(v))$
    \State precommit $\id(v)$
  \Else
    \State precommit $\nil$
  \EndIf
  \State $\step_p \gets \pcstep$
\EndUpon
\Upon{$2f+1$ precommits for $\id(x)$ in round $\round_p$ with $\round_p > \validRound_p$} \label{ln:obs}
  \State $\validPair{p} \gets (\round_p,\id(x))$
\EndUpon
\Upon{\new{a quorum of round $\round_p$, for the first time in this round}} \label{ln:exit}
  \State \new{\textbf{if} $\step_p = \propose$ \textbf{then} precommit $y$; $\step_p \gets \pcstep$, where $y = \id(x)$ for a value $x$ with $2f+1$ precommits in the quorum, chosen deterministically if several, and $y = \nil$ if there is none} \Comment{exit precommit} \label{ln:exitvote}
  \State \new{\textbf{if} $\decision_p = \nil$ \textbf{then} $\mathit{StartRound}(\round_p+1)$} \label{ln:exitstart}
\EndUpon
\Upon{$n-f$ precommits $\langle\PRECOMMIT, s, \id(v)\rangle$ from distinct processes, for some round $s$, and $v$ is known} \label{ln:decide}
  \State $\decision_p \gets v$
\EndUpon
\Function{OnTimeoutPropose}{$r$}
  \State \textbf{if} $r = \round_p \wedge \step_p = \propose$ \textbf{then} precommit $\nil$; $\step_p \gets \pcstep$
\EndFunction
\end{algorithmic}
\end{algorithm}

Walk restricts the quorum rule to the current round and removes the precommit timer and its handler. If a process has not yet precommitted when it leaves, it first precommits a value with $2f+1$ precommits in the quorum, choosing deterministically if several qualify. If no value qualifies, it precommits $\nil$.

This \emph{exit precommit} is an ordinary $\PRECOMMIT$. A non-nil exit precommit requires a support certificate but not the full value. It is cast in the current round, so the existing signer restriction to monotone (height, round) pairs is preserved. The observation rule (line~\ref{ln:obs}) is unchanged.

The proposer runs $\mathit{WaitForValid}$ with a separate timeout $\twait(r)$, adopting only certificates from rounds up to $r-1$. The observation rule handles a certificate of the current round; messages from later rounds remain buffered until those rounds are reached. A decision requires $n-f$ matching precommits from one round and the full value, which may come from any source.

\paragraph{Agreement.}
Suppose $v$ obtains a decision certificate in round $r$. Every quorum of round $r$ contains at least $n-3f \geq 2f+1$ precommits for $v$ and fewer than $2f+1$ precommits for any other value. Observation takes priority over exit, so every correct process leaving $r$ records support for $v$ first.

After leaving $r$, a correct process can accept a conflicting value only with a sufficiently recent certificate for that value. But forming the first such certificate would require correct processes to have already precommitted a conflicting value. The proof formalizes this dependence by proving three invariants together, covering the recorded state, later precommits and later certificates.

\paragraph{Precommitting before leaving a round.}
Without the post-quorum wait, a process may receive a quorum and leave before the proposal arrives. Its missing precommit can prevent a decision even when the proposal is valid. The exit precommit ensures that the process participates before leaving: it precommits a value supported in the quorum, or $\nil$ if no value has enough support. It still signs only in its current round. The proposal timeout handles a silent proposer, while $\mathit{WaitForValid}$ gives a correct proposer time to learn the certificates recorded by other processes.

\paragraph{Guarantees.}
At $n\geq5f+1$ and at most $f$ Byzantine processes, Walk satisfies Agreement, Validity and Termination (Appendix~\ref{app:walk}); Agreement and Validity use no timing assumption. The optional all-nil shortcut is disabled in every termination result here; Appendix~\ref{sec:shortcut} discusses its unresolved termination argument.

\section{Justify}
\label{sec:justify}

In Walk, each process keeps a private valid state, and the proposer may wait up to $\twait$ to learn the certificates recorded by other processes. Justify carries the evidence needed for acceptance in the proposal itself. It removes the recorded valid state and $\mathit{WaitForValid}$, at the cost of larger proposals.

\begin{definition}
\label{def:just}
A \emph{justification} for round $s$ is a triple $J = (vr, C, D)$ with $-1 \leq vr < s$, where $C$ is a certificate for some value identifier $x_J$ in round $vr$ if $vr \geq 0$ and $C = \bot$ if $vr = -1$, and $D$ contains one dead certificate for every round $j$ with $vr < j < s$.
\end{definition}

Algorithm~\ref{alg:justify} gives Justify in full. It removes Walk's recorded state variables $\validRound_p$ and $\validValueId_p$, $\mathit{WaitForValid}$ and the observation rule. The exit rule, decision rule and proposal timeout are unchanged.

To build a justification for round $r$, the proposer selects the largest round $vr<r$ for which it holds a support certificate. It chooses any such certificate $C$ if several exist at $vr$. If it holds none, it uses $vr=-1$ and $C=\bot$. For each intervening round $j$ with $vr<j<r$, it adds to $D$ the quorum on which it left $j$. That quorum contains no support certificate, since otherwise the proposer would have selected a larger $vr$. It is therefore a dead certificate by Lemma~\ref{lem:exit}.

The proposal is $\langle\PROPOSAL,r,x,J\rangle$. It carries a fresh full value $x$ when $vr=-1$, and the certified identifier $x_J$ otherwise. A recipient accepts only if $J$ is a justification for its current round and the payload passes the corresponding check: $\valid(x)$ for a fresh value, or $x=x_J$ for a certified identifier. This single acceptance rule replaces Walk's fresh-proposal and reproposal rules.

\begin{algorithm}[htp]
\caption{Justify, one height, process $p$. Blue lines differ from Algorithm~\ref{alg:walk}. A proposal carries its justification; there is no recorded valid state and no $\mathit{WaitForValid}$. Fresh proposals carry a full value; certified reproposals carry an identifier, which is precommitted directly. Rule~\ref{ln:jexit} applies at any time. Messages of rounds above $\round_p$ are kept and take effect when $\round_p$ reaches their round.}
\label{alg:justify}
\small
\begin{algorithmic}[1]
\State $\round_p := 0$; $\step_p := \propose$; $\decision_p := \nil$
\Upon{start}
  \State $\mathit{StartRound}(0)$
\EndUpon
\Function{StartRound}{$r$}
  \State $\round_p \gets r$; $\step_p \gets \propose$
  \If{$\proposer(r) = p$}
    \State \new{$vr \gets$ the largest round $j<r$ with a held certificate $C_j$, or $-1$; choose any $C_j$ if several; $C_{-1}:=\bot$}
    \State \new{$D \gets$ for each round $j$ with $vr < j < r$, the quorum on which $p$ left round $j$}
    \State \new{$x \gets$ the identifier certified by $C_{vr}$ \textbf{if} $vr \geq 0$ \textbf{else} $\mathit{getValue}()$}
    \State \textbf{if} $\round_p\neq r\vee\step_p\neq\propose\vee\decision_p\neq\nil$ \textbf{then return}
    \State \new{broadcast $\langle\PROPOSAL, r, x, (vr, C_{vr}, D)\rangle$} \label{ln:jbroadcast}
  \Else
    \State schedule $\mathit{OnTimeoutPropose}(r)$ after $\tprop(r)$
  \EndIf
\EndFunction
\Upon{\new{$\langle\PROPOSAL, \round_p, x, J\rangle$ from $\proposer(\round_p)$ while $\step_p = \propose$}} \label{ln:jprop}
  \State \new{\textbf{if} $J$ is a justification for round $\round_p$, $J.vr=-1$, and $x$ is a full value with $\valid(x)$ \textbf{then} precommit $\id(x)$}
  \State \new{\textbf{else if} $J$ is a justification for round $\round_p$, $J.vr\geq0$, and $x=x_J$ \textbf{then} precommit $x$ \textbf{else} precommit $\nil$}
  \State $\step_p \gets \pcstep$
\EndUpon
\Upon{a quorum of round $\round_p$, for the first time in this round} \label{ln:jexit}
  \State \textbf{if} $\step_p = \propose$ \textbf{then} precommit $y$; $\step_p \gets \pcstep$, where $y = \id(x)$ for a value $x$ with $2f+1$ precommits in the quorum, chosen deterministically if several, and $y = \nil$ if there is none \Comment{exit precommit}
  \State \textbf{if} $\decision_p = \nil$ \textbf{then} $\mathit{StartRound}(\round_p+1)$
\EndUpon
\Upon{$n-f$ precommits $\langle\PRECOMMIT, s, \id(v)\rangle$ from distinct processes, for some round $s$, and $v$ is known} \label{ln:jdecide}
  \State $\decision_p \gets v$
\EndUpon
\Function{OnTimeoutPropose}{$r$}
  \State \textbf{if} $r = \round_p \wedge \step_p = \propose$ \textbf{then} precommit $\nil$; $\step_p \gets \pcstep$
\EndFunction
\end{algorithmic}
\end{algorithm}

A correct proposer can build $J$ as soon as it enters round $r$. It has retained the quorum on which it left each earlier round. After it selects $vr$, each retained exit quorum for $vr<j<r$ supplies the required dead certificate. No additional messages are needed to construct the justification.

\paragraph{Agreement from the justification.}
Suppose $v$ has a decision certificate in round $r$. Any justification for a later round must have $vr\geq r$: choosing $vr<r$ would require a dead certificate for round $r$, which cannot coexist with the decision certificate. If $vr=r$, its support certificate is for $v$. If $vr>r$, a certificate for a conflicting value would require earlier correct precommits for that value in rounds greater than $r$. Tracing the first conflicting precommit back to its justification, or to earlier matching precommits for an exit vote, gives the same contradiction. The full proof appears in Appendix~\ref{app:justify}.

The recipient can therefore check the proposal using its attached justification. Acceptance does not depend on the recipient's recorded valid state.

\paragraph{Proposing without a synchronization wait.}
A correct proposer already holds the evidence needed to construct its justification. It can broadcast immediately, and recipients can check that evidence even if they have learned different certificates. The exit precommit ensures that a process receiving a quorum before the proposal still participates before advancing. Justify retains $\tprop$ for silent proposers, but needs neither $\mathit{WaitForValid}$ nor the post-quorum wait.

\paragraph{Correctness, latency and message size.}
At $n \geq 5f+1$ and at most $f$ Byzantine processes, Justify satisfies Agreement, Validity and Termination (Appendix~\ref{app:justify}); Agreement and Validity use no timing assumption.

Consider a round $r$ with a correct proposer and $\tprop(r)>2\Delta$, whose first correct entry occurs after $\GST$. Every correct process decides within $3\Delta$ of that first entry. If the proposer enters the round before deciding, it broadcasts upon entry and every correct process decides within $2\Delta$ of the broadcast. These are the two communication steps of Fast Tendermint; the additional $\Delta$ in the first-entry bound allows for the proposer's later entry (Theorem~\ref{thm:responsive}).

A proposal carries at most one support certificate and one dead certificate per intervening round. A correct proposer uses its exit quorums for rounds strictly between $vr$ and $r$. A process already in one of those rounds can use the attached quorums to catch up; a process at or below $vr$ still needs the missing round quorums. An arbitrary valid justification may contain smaller dead certificates. Those suffice to justify acceptance, but may not suffice for quorum-based catch-up.

\section{Choosing a construction}
\label{sec:discussion}

Walk is the natural starting point when preserving Tendermint's local valid state and existing proposal format matters most. Justify is the alternative when eliminating the proposer's synchronization wait is worth carrying more evidence. Both retain per-height decisions and ordinary precommits, so the choice does not require adopting a chained finality interface.

\paragraph{Advancing on fewer precommits.}
The main constructions retain the published $n-f$ exit threshold. Appendix~\ref{app:fastexit} separately shows how Walk and Justify can advance on $2f+1$ matching precommits, or on a dead certificate that may be larger. This optional extension changes neither the decision threshold nor the single-precommit signing rule. Its relation to Minimmit's notarizations and nullifications is discussed there, so the main contribution remains the round-change repair.

\paragraph{Extension to chains.}
In the per-height constructions, a proposer that selects a certified value reproposes it rather than extending it. Appendix~\ref{sec:chains} instead uses the certified block as the parent of a new block. In this extension of Walk and Justify, Agreement requires the committed sequences at correct processes to be prefixes of a single chain, as in the use of notarizations by Minimmit and Simplex. The safety argument uses the property shared by both repairs: a correct process holds exit evidence for every round it passes. Continuing state-machine replication also requires the data-availability and execution conditions stated in that appendix.

\paragraph{Other round-change designs.}
Appendix~\ref{app:alternatives} sketches three other ways to retain round skipping: broaden the observation rule, jump on a support certificate after a timer, or carry recorded certificates in precommits. These sketches explain the design choices; Walk and Justify remain the two constructions specified and proved here.

\section{Conclusion}

The published protocol allows processes to advance past a round without recording support for a value that may have been decided there. Walk repairs this by requiring evidence from each round and a precommit before leaving on the current round's quorum. Justify carries the evidence in the proposal and removes the proposer's synchronization wait. Both preserve Agreement and Validity under asynchrony and Termination under the original partially synchronous gossip model. Both retain the proposal timeout.

\section*{Acknowledgments}
We thank Daniel Cason and Preston Vander Vos for reproducing the execution of Theorem~\ref{thm:violation} by modeling the published pseudocode in Quint.

\clearpage
\appendix
\numberwithin{theorem}{section}

\section{The published algorithm}
\label{app:ft}

Algorithm~\ref{alg:ft} restates Algorithm~1 of~\cite{ft} in our notation, line for line, with the original line numbers. The published pair $\mathit{valid}_p = (\round, \id)$ is written $\validPair{p}$, and the published predicate $\mathit{validate}$ is written $\valid$.

\setcounter{algorithm}{0}
\begin{algorithm}[htp]
\caption{Algorithm~1 of~\cite{ft}, restated with its original line numbers. Fresh proposals carry a full value $v$; reproposals carry an identifier $\id(v)$.}
\label{alg:ft}
\footnotesize
\begin{algorithmic}[1]
\State $\round_p := 0$ \label{ft:1}
\State $\step_p \in \{\propose, \pcstep\}$ \label{ft:2}
\State $\decision_p := \nil$ \label{ft:3}
\State $\validPair{p} := (-1, \nil)$ \label{ft:4}
\Function{StartRound}{$\round$} \label{ft:5}
  \State $\round_p \gets \round$ \label{ft:6}
  \State $\step_p \gets \propose$ \label{ft:7}
  \If{$\proposer(\round_p) = p$} \label{ft:8}
    \If{$\round_p > 0$} \label{ft:9}
      \State $\mathit{WaitForValid}(\tprec(\round_p))$ \label{ft:10}
    \EndIf
    \If{$\validValueId_p = \nil$} \label{ft:11}
      \State $\mathit{proposal} \gets \mathit{getValue}()$ \label{ft:12}
    \Else \label{ft:13}
      \State $\mathit{proposal} \gets \validValueId_p$ \label{ft:14}
    \EndIf
    \State broadcast $\langle\PROPOSAL, \round_p, \mathit{proposal}, \validRound_p\rangle$ \label{ft:15}
  \Else \label{ft:16}
    \State schedule $\mathit{OnTimeoutPropose}(\round_p)$ after $\tprop(\round_p)$ \label{ft:17}
  \EndIf
\EndFunction
\Upon{$\langle\PROPOSAL, \round_p, v, -1\rangle$ from $\proposer(\round_p)$ while $\step_p = \propose$} \label{ft:18}
  \If{$\valid(v) \wedge (\validRound_p = -1 \vee \validValueId_p = \id(v))$} \label{ft:19}
    \State broadcast $\langle\PRECOMMIT, \round_p, \id(v)\rangle$ \label{ft:20}
  \Else \label{ft:21}
    \State broadcast $\langle\PRECOMMIT, \round_p, \nil\rangle$ \label{ft:22}
  \EndIf
  \State $\step_p \gets \pcstep$ \label{ft:23}
\EndUpon
\Upon{$\langle\PROPOSAL, \round_p, \id(v), vr\rangle$ from $\proposer(\round_p)$ with $0 \leq vr < \round_p$ and $2f+1$ $\langle\PRECOMMIT, vr, \id(v)\rangle$ while $\step_p = \propose$} \label{ft:24}
  \If{$\validRound_p \leq vr \vee \validValueId_p = \id(v)$} \label{ft:25}
    \If{$\validRound_p \leq vr$} \label{ft:26}
      \State $\validPair{p} \gets (vr, \id(v))$ \label{ft:27}
    \EndIf
    \State broadcast $\langle\PRECOMMIT, \round_p, \id(v)\rangle$ \label{ft:28}
  \Else \label{ft:29}
    \State broadcast $\langle\PRECOMMIT, \round_p, \nil\rangle$ \label{ft:30}
  \EndIf
  \State $\step_p \gets \pcstep$ \label{ft:31}
\EndUpon
\Upon{$2f+1$ $\langle\PRECOMMIT, \round_p, \id(v)\rangle$ with $\round_p > \validRound_p$} \label{ft:32}
  \State $\validPair{p} \gets (\round_p, \id(v))$ \label{ft:33}
\EndUpon
\Upon{$n-f$ $\langle\PRECOMMIT, r, \ast\rangle$ for the first time with $r \geq \round_p$} \label{ft:34}
  \State schedule $\mathit{OnTimeoutPrecommit}(r)$ after $\tprec(r)$ \label{ft:35}
\EndUpon
\Upon{$\langle\PROPOSAL, r, v, -1\rangle$ from $\proposer(r)$ and $n-f$ $\langle\PRECOMMIT, r', \id(v)\rangle$} \label{ft:36}
  \State $\decision_p \gets v$ \label{ft:37}
\EndUpon
\Function{WaitForValid}{$\mathit{timeout}$} \Comment{run by the proposer of rounds $> 0$} \label{ft:38}
  \While{$\validRound_p < \round_p - 1 \wedge \neg\, \mathit{timeout.elapsed}()$} \label{ft:39}
    \If{$2f+1$ $\langle\PRECOMMIT, r, \id(v)\rangle$ with $r > \validRound_p$} \label{ft:40}
      \State $\validPair{p} \gets (r, \id(v))$ \label{ft:41}
    \EndIf
  \EndWhile
\EndFunction
\Function{OnTimeoutPropose}{$\round$} \label{ft:42}
  \If{$\round = \round_p \wedge \step_p = \propose$} \label{ft:43}
    \State broadcast $\langle\PRECOMMIT, \round_p, \nil\rangle$ \label{ft:44}
    \State $\step_p \gets \pcstep$ \label{ft:45}
  \EndIf
\EndFunction
\Function{OnTimeoutPrecommit}{$\round$} \label{ft:46}
  \If{$\round \geq \round_p \wedge \decision_p = \nil$} \label{ft:47}
    \State $\mathit{StartRound}(\round + 1)$ \label{ft:48}
  \EndIf
\EndFunction
\end{algorithmic}
\end{algorithm}

\section{Evidence lemmas}
\label{app:evidence}

\begin{lemma}
\label{lem:evidence}
Let $T$ be a set of precommits of round $s$.
\begin{enumerate}[nosep,label=(\roman*)]
\item If $T$ contains a certificate for $x$, then no value other than $x$ obtains $n-f$ precommits from distinct processes in round $s$.
\item If $T$ is a dead certificate, then no value obtains $n-f$ precommits from distinct processes in round $s$.
\item If $T$ is a quorum all of whose precommits are for $\nil$, then no value obtains $2f+1$ precommits from distinct processes in round $s$.
\item If some value $v$ obtains $n-f$ precommits from distinct processes in round $s$, then no value other than $v$ has a certificate in round $s$, and no dead certificate for round $s$ exists.
\end{enumerate}
\end{lemma}

\begin{proof}
(i) Any $2f+1$ processes and any $n-f$ processes share at least $(2f+1)+(n-f)-n = f+1$ processes, one of which is correct; a correct process precommits once per round.

(ii) Fix $x$. At least $2f+1$ senders in $T$ precommit something other than $\id(x)$ in round $s$, at least $f+1$ of them correct, and a correct process sends one precommit per round, so those never precommit $\id(x)$ in round $s$; hence at most $n-f-1$ processes precommit $\id(x)$.

(iii) A process with a precommit for $\id(x)$ in round $s$ is outside $T$, and there are $n-(n-f) = f$ such processes, or it is a Byzantine member of $T$ that also sent $\nil$, and there are at most $f$ of those; so at most $2f$.

(iv) The first claim is (i). For the second, a dead certificate contradicts (ii).

Part (iii) does not extend to a quorum in which some value has a few precommits: a Byzantine member of $T$ that precommitted $\nil$ in $T$ may also have precommitted $\id(x)$, so $x$ can have up to $a_x(T) + 2f$ precommits, and with $a_x(T) = 1$ a certificate is possible.
\end{proof}

\begin{lemma}[round progress]
\label{lem:exit}
Let $n \geq 5f+1$. If at least $n-f$ distinct processes have precommitted in round $s$, then the set of their precommits contains a certificate for some value or is a dead certificate.
\end{lemma}

\begin{proof}
Take one precommit per sender. If no value has $2f+1$ of them, then for every $x$ at least $n-f-2f \geq 2f+1$ of the senders precommit something other than $\id(x)$.
\end{proof}

\section{Correctness of Walk}
\label{app:walk}

\subsection{Agreement and Validity}

Three properties of Algorithm~\ref{alg:walk} are used throughout. First, $\validRound_p \leq \round_p$ while $p$ has not decided, and $\validRound_p = \round_p$ only through Rule~\ref{ln:obs} with a certificate of the current round: line~\ref{ln:wfv} sets $\validRound_p$ to at most $r-1$ and Rule~\ref{ln:repro} sets it to $vr < \round_p$.

Second, every non-nil $\validValueId_p$ is backed by a certificate held by $p$ in round $\validRound_p$.

Third, after initialization, $\mathit{StartRound}$ is called only by Rule~\ref{ln:exit}, with $\round_p+1$, so a process reaches a round above $r$ only by leaving round $r$ on a quorum of round $r$.

\begin{theorem}
\label{thm:walk-safe}
Let $n \geq 5f+1$ and at most $f$ Byzantine processes. Walk satisfies Agreement.
\end{theorem}

\begin{proof}
Fix an execution and let $v$ be a value that, at some point, has received precommits from a set $V$ of $n-f$ distinct processes in round $r$.

Counting in round $r$. The facts below concern the set of all precommits of round $r$ ever sent, so they hold throughout the execution. A process with a precommit for $\id(x) \neq \id(v)$ in round $r$ is either outside $V$, and there are $f$ such processes, or a Byzantine member of $V$, and there are at most $f$ of those. So $x$ has at most $2f$ precommits from distinct processes in round $r$: no certificate for $x \neq v$ ever exists in round $r$.

By the same count, any quorum of round $r$ has at most $2f$ members whose precommit is not for $\id(v)$, hence at least $n-3f \geq 2f+1$ precommits for $\id(v)$: every quorum of round $r$ contains a certificate for $v$, and no quorum of round $r$ contains $2f+1$ precommits for another value.

Consider the following three statements about a prefix of the execution.
\begin{enumerate}[nosep,label=(I\arabic*)]
\item Every correct process $p$ with $\round_p>r$ that has not decided has $\validRound_p\geq r$ and $\validValueId_p=\id(v)$.
\item No correct process has sent $\langle\PRECOMMIT, s, \id(x)\rangle$ with $s > r$ and $x \neq v$.
\item For every round $s \geq r$ and every $x \neq v$, fewer than $2f+1$ distinct processes have sent $\langle\PRECOMMIT, s, \id(x)\rangle$.
\end{enumerate}
All three hold at the start of the execution, when every process is in round $0 \leq r$ and nothing has been sent. We show that no event is the first to break any of them. Events are the sends of all processes and the state changes of correct processes, ordered by time.

\emph{First violation of (I2).} A correct $p$ with $\round_p = s > r$ sends a precommit for $\id(x)$, $x \neq v$. By (I1), $\validPair{p} = (s',\id(v))$ with $s' \geq r$.

Rule~\ref{ln:fresh} needs $\validRound_p = -1$ or $\validValueId_p = \id(x)$; both fail. Rule~\ref{ln:repro} needs a certificate for $x$ in a round $vr$ and either $\validValueId_p = \id(x)$, which fails, or $\validRound_p \leq vr$, which gives $vr \geq r$ and a certificate for $x$ in a round at or above $r$, contradicting (I3).

An exit precommit for $\id(x)$ (line~\ref{ln:exitvote}) needs $2f+1$ precommits for $\id(x)$ in a quorum of round $s \geq r$, that is, a certificate for $x$ in round $s$, contradicting (I3). $\mathit{OnTimeoutPropose}$ sends $\nil$.

\emph{First violation of (I3).} Some process becomes the $(2f+1)$-th distinct sender of a precommit for $\id(x)$ in a round $s \geq r$, $x \neq v$. If $s = r$ this contradicts the counting. If $s > r$, at least $f+1$ of the $2f+1$ senders are correct. If the send of such a correct sender is an earlier event, it contradicts (I2), which holds before this event; if it is this very event, the event breaks (I2), which the previous case excludes.

\emph{First violation of (I1).} Either a correct $p$ moves to a round above $r$ without satisfying both $\validRound_p\geq r$ and $\validValueId_p=\id(v)$, or such a $p$ updates its valid state so that one of these conditions fails. By the third property, $p$ moves above $r$ by leaving round $r$ on a quorum $Q$ of round $r$ through Rule~\ref{ln:exit}. $Q$ contains a certificate for $v$, so by the time $p$ holds $Q$ it holds $2f+1$ precommits for $\id(v)$ in its current round $r$, and Rule~\ref{ln:obs}, which is applied before Rule~\ref{ln:exit}, has set $\validPair{p} \gets (r,\id(v))$ unless $\validRound_p = r$ already; in that case the valid state was set by Rule~\ref{ln:obs} from a certificate of round $r$, which is for $v$. So $p$ leaves round $r$ with $\validPair{p} = (r, \id(v))$.

An update to the valid state at a process above $r$ comes from Rule~\ref{ln:obs} with a certificate for $x$ in round $\round_p > r$, from Rule~\ref{ln:repro} with a certificate in a round $vr \geq \validRound_p \geq r$, or from line~\ref{ln:wfv} with a certificate in a round above $\validRound_p \geq r$; by (I3), which holds before the event, each of these certificates is for $v$.

Thus (I1) to (I3) hold throughout the execution. If a correct process decides $w$ on $n-f$ precommits of a round $s \geq r$, then $n-f \geq 2f+1$ distinct processes precommitted $\id(w)$ in round $s$, and (I3) gives $w = v$. If $s < r$, exchange the roles of the two decisions.
\end{proof}

\begin{theorem}
\label{thm:walk-valid}
Let at most $f$ processes be Byzantine. Walk satisfies Validity.
\end{theorem}

\begin{proof}
Let $v$ be decided and consider the earliest precommit for $\id(v)$ sent by a correct process. It was not sent by Rule~\ref{ln:repro}, which needs a certificate for $v$ in an earlier round and hence an earlier correct precommit for $\id(v)$. It was not an exit precommit, which needs $2f+1$ precommits for $\id(v)$ in a quorum of the same round and hence $f+1$ earlier correct precommits for $\id(v)$. So it was sent by Rule~\ref{ln:fresh}, which checked $\valid(v)$. A decision needs $n-f \geq f+1$ precommits for $\id(v)$, at least one from a correct process.
\end{proof}

\subsection{Termination}
\label{sec:walk-termination}

Let $\tprop(r)$ and $\twait(r)$ be finite, positive and nondecreasing in $r$. As in~\cite{ft}, assume that every sufficiently large round $r$ satisfies
\[
\twait(r)>2\Delta,\qquad \tprop(r)>\twait(r)+2\Delta.
\] We write ``a process is in round $s$'' for a correct process that has not decided and has $\round_p = s$.

\begin{lemma}[catch-up]
\label{lem:catchup}
After $\GST$, if a correct process $p$ is in round $s$ at time $t$, then every correct process that has not decided is in a round at least $s$ by time $t+\Delta$.
\end{lemma}

\begin{proof}
$p$ passed every round $j < s$ by Rule~\ref{ln:exit}, on a quorum $E_j$ of round $j$ it had received by $t$. By the gossip property every correct process $q$ holds every $E_j$ by $t+\Delta$. If $q$ is in a round $j < s$ at that time, then $E_j$ is a quorum of its current round, so Rule~\ref{ln:exit} fires and $q$ moves to round $j+1$; the rule applies during $\mathit{WaitForValid}$ too. Repeating, $q$ reaches round $s$. Processing takes no time in the model.
\end{proof}

\begin{lemma}[round-change synchrony]
\label{lem:sync}
After $\GST$, if a correct process leaves round $s$ at time $t$, then every correct process that has not decided is in a round above $s$ by $t+\Delta$.
\end{lemma}

\begin{proof}
Lemma~\ref{lem:catchup} with $s+1$.
\end{proof}

\begin{lemma}[every round ends]
\label{lem:ends}
No correct process remains undecided in a fixed round forever.
\end{lemma}

\begin{proof}
If a correct process decides, its decision evidence and value reach every correct process by the gossip bound. Otherwise suppose a correct process remains in round $s$ forever. Catch-up places every correct process in a round at least $s$. If any process leaves $s$, round-change synchrony makes the remaining processes leave too, a contradiction. Thus all correct processes remain in $s$.

Each nonproposer eventually precommits on a proposal or its proposal timeout. If the proposer is correct, its finite $\mathit{WaitForValid}$ completes and it broadcasts and processes its own proposal.

Therefore every correct process precommits in $s$. Their at least $n-f$ precommits form a quorum of round $s$ that reaches every correct process, and Rule~\ref{ln:exit} makes every correct process leave $s$, again a contradiction.
\end{proof}

\begin{lemma}[proposer]
\label{lem:proposer}
Let $p = \proposer(r)$ be correct, $r > 0$, let $t_1 > \GST$ be the time at which the first correct process enters round $r$, and let $\twait(r) > 2\Delta$. Let $p$ broadcast a proposal for $x$ with advertised round $vr$ at time $t_p$. Then $t_p \leq t_1 + \Delta + \twait(r)$; every correct process $q$ that is in round $r$ at a time $t \geq t_p$ and has not yet precommitted in round $r$ has $\validRound_q \leq vr$ or $\validRound_q = r$ at time $t$; and if $vr \geq 0$ every correct process holds the certificate for $x$ in round $vr$ by $t_p + \Delta$.
\end{lemma}

\begin{proof}
By Lemma~\ref{lem:catchup}, $p$ enters round $r$ by $t_1 + \Delta$, and $\mathit{WaitForValid}$ lasts at most $\twait(r)$; this gives the bound on $t_p$.

Let $q$ be a correct process in round $r$ at time $t\geq t_p$ that has not yet precommitted there, and suppose $\validRound_q\neq r$. For $q=p$, its unvoted state is the one used at the broadcast, so $\validRound_q=vr$: later writes in $r$ either set the round to $r$ or accompany its precommit.

Now take $q\neq p$. While $q$ is in $r$, Rule~\ref{ln:obs} sets its recorded round to $r$, Rule~\ref{ln:repro} accompanies its precommit, and it does not run the proposer's wait. Its recorded certificate was therefore received while in a round at most $r-1$, hence by $t_1 + \Delta$ since $q$ entered round $r$ by then, and by the gossip property $p$ holds it by $t_1 + 2\Delta$.

If $\mathit{WaitForValid}$ returned because $\validRound_p = r-1$, then $vr = r-1 \geq \validRound_q$, since $\validRound_q \leq r-1$. Otherwise it returned at $t_p > t_1 + 2\Delta$, when $p$ held every certificate $q$ had recorded, and line~\ref{ln:wfv} had raised $\validRound_p$ to the highest round among the certificates $p$ holds of rounds up to $r-1$; so again $vr \geq \validRound_q$.

Finally $p$ received the certificate for $x$ in round $vr$ by $t_p$, so every correct process has it by $t_p + \Delta$.
\end{proof}

\begin{lemma}[good round]
\label{lem:good}
Let $r>0$ be a round with a correct proposer $p$, let $t_1 > \GST$ be the time at which the first correct process enters round $r$, and let $\twait(r) > 2\Delta$ and $\tprop(r) > \twait(r) + 2\Delta$. Then every correct process decides by $t_1 + 3\Delta + \twait(r)$.
\end{lemma}

\begin{proof}
Put $T=t_1+\Delta+\twait(r)$, the latest time at which the correct proposer could finish an uninterrupted wait. If some correct process decides by $T+\Delta$, gossip makes every correct process decide by $T+2\Delta$, the claimed bound. Otherwise assume no such decision.

First, $p$ broadcasts by $T$. Catch-up puts it in round $r$ by $t_1+\Delta$ unless it has already left. Before its broadcast and up to $T$, there can be no first correct precommit in round $r$: a proposal vote requires its proposal, a timeout fires no earlier than $t_1+\tprop(r)>T+\Delta$, and an exit precommit requires a quorum of round $r$, which needs at least $n-2f$ correct precommits in round $r$. With only Byzantine precommits, no quorum of round $r$ exists. Thus no correct process leaves $r$ before the broadcast, the proposer's wait is not aborted, and it broadcasts a proposal for $x$ with advertised round $vr$ at some $t_p\leq T$.

Let $H=t_p+\Delta$. Up to $H$, every correct precommit in round $r$ is for $\id(x)$. Consider the first correct precommit in round $r$ that is not. It is not a timeout, which is later than $H$.

It is not an exit precommit: a quorum of round $r$ holds at least $n-2f \geq 2f+1$ correct precommits, all for $\id(x)$ before this event, and no other value has $2f+1$ precommits in it, so the exit precommit is for $\id(x)$.

It is not a rejection of $p$'s proposal: by Lemma~\ref{lem:proposer}, the rejecting process $q$ has $\validRound_q \leq vr$, in which case the proposal rule accepts, with $\valid(x)$ holding for a fresh proposal since $p$ is correct, or $\validRound_q = r$, in which case $\validValueId_q$ was recorded by Rule~\ref{ln:obs} from $2f+1$ precommits in round $r$, of which at least $f+1$ are correct and hence for $\id(x)$, so $\validValueId_q = \id(x)$ and the proposal rule accepts by value equality.

Catch-up puts all correct processes in a round at least $r$ by $t_1+\Delta\leq H$, and the proposal and its supporting certificate reach them by $H$. Each correct process has consequently precommitted $\id(x)$ in $r$ by $H$: on receiving the proposal, or as an exit precommit on a quorum of round $r$, which is for $\id(x)$ as just shown.

Every correct process holds these at least $n-f$ precommits by $H+\Delta$. It also holds $x$: a fresh value travels in $p$'s proposal; otherwise the certificate used by $p$ contains a correct precommit sent before $t_p$, and the earliest correct precommit for that identifier followed a validated fresh proposal, by Theorem~\ref{thm:walk-valid}. Gossip delivers that proposal by $t_p+\Delta$. All processes decide by $t_p+2\Delta\leq T+2\Delta=t_1+3\Delta+\twait(r)$.
\end{proof}

\begin{theorem}
\label{thm:termination}
Let $n \geq 5f+1$ and at most $f$ Byzantine processes. Walk satisfies Termination.
\end{theorem}

\begin{proof}
If some correct process decides $v$ at time $t$, every correct process holds the same $n-f$ precommits and the value by $\max\{t,\GST\} + \Delta$ and decides. Otherwise, by Lemma~\ref{lem:ends}, rounds advance forever after $\GST$. Proposers rotate, so infinitely many rounds have correct proposers; the eventual timeout assumptions ensure that from some round on they satisfy the conditions of Lemma~\ref{lem:good}; and from some round on the first correct process enters the round after $\GST$. The first round that meets all three decides at every correct process.
\end{proof}

\section{Correctness of Justify}
\label{app:justify}

In the proofs below, $x$ denotes the underlying proposed value: the proposal carries $x$ when fresh and $\id(x)$ when certified. The precommit is for $\id(x)$ in both cases; possession of the full value is required for deciding, not for a certified reproposal.

\begin{theorem}
\label{thm:justify-safe}
Let $n \geq 5f+1$ and at most $f$ Byzantine processes. Justify satisfies Agreement and Validity.
\end{theorem}

\begin{proof}
Let $v$ obtain precommits from a set $V$ of $n-f$ distinct processes in round $r$. By Lemma~\ref{lem:evidence}(iv), no value other than $v$ has a certificate in round $r$, and no dead certificate for round $r$ exists.

We claim that no correct process ever sends a precommit for $\id(x)$ with $x \neq v$ in a round above $r$. Suppose otherwise and consider the earliest such precommit, by $p$ in round $s > r$. It is not an exit precommit, since that needs $2f+1$ precommits for $\id(x)$ in a quorum of round $s$, at least $f+1$ of them from correct processes that precommitted $\id(x)$ in round $s$ earlier. Therefore $p$ accepted a proposal for $x$ with $J=(vr,C,D)$ a justification for round $s$.

If $vr < r$, then $D$ contains a dead certificate for round $r$, which does not exist. If $vr = r$, then $C$ is a certificate for $x$ in round $r$, so $x = v$. If $r < vr < s$, then $C$ is a certificate for $x$ in round $vr$, at least $f+1$ of whose senders are correct processes that sent a precommit for $\id(x)$ in round $vr > r$ earlier than $p$'s, contradicting the choice of $p$.

Now let a correct process decide $w$ on $n-f$ precommits of a round $s$. If $s > r$, at least one of these precommits comes from a correct process, so $w = v$ by the claim. If $s = r$, then $w = v$ by the round-$r$ counting. If $s < r$, exchange the roles.

Validity: the earliest correct precommit for a decided value is an acceptance with $vr = -1$, which checks $\valid$; acceptances with $vr \geq 0$ and exit precommits need earlier correct precommits for the same identifier.
\end{proof}

Acceptance depends on the proposal rather than on the recipient's recorded state. A lagging process that holds a certificate for an old round can therefore accept a fresh proposal carrying a dead certificate for that round. This removes the wait at the proposer.

\begin{theorem}
\label{thm:responsive}
Let $n \geq 5f+1$ and at most $f$ Byzantine processes, and let $\tprop(r)$ be finite, positive, nondecreasing and greater than $2\Delta$ for every sufficiently large $r$. Justify satisfies Termination. Moreover, let $r$ be a round with a correct proposer $p$, let $t_1 > \GST$ be the time at which the first correct process enters it, and let $\tprop(r) > 2\Delta$. Then every correct process decides by $t_1 + 3\Delta$; and if $p$ enters round $r$ before deciding, at a time $t_p \leq t_1 + \Delta$, it broadcasts its proposal at $t_p$ and every correct process decides by $t_p + 2\Delta$. If after $\GST$ the dissemination property holds with an actual delay $\delta \leq \Delta$, the same argument gives $t_p + 2\delta$ and $t_1 + 3\delta$.
\end{theorem}

\begin{proof}
Lemmas~\ref{lem:catchup}, \ref{lem:sync} and~\ref{lem:ends} hold for Justify, since its exit rule is that of Walk and its proposer broadcasts without waiting.

If $p$ decides before entering round $r$, it does so by $t_1 + \Delta$, since otherwise catch-up would have put it in round $r$ by then; gossip then makes every correct process decide by $t_1 + 2\Delta$.

Otherwise $p$ enters round $r$ at some $t_p \in [t_1, t_1+\Delta]$, holding a justification as argued in Section~\ref{sec:justify}, and broadcasts its proposal for $x$ with justification $J$ at $t_p$.

If some correct process decides by $t_p + \Delta$, gossip makes every correct process decide by $t_p + 2\Delta$. Otherwise every correct process is still active until $t_p + \Delta$, and we show that all of them precommit $\id(x)$ in round $r$ by then.

Every correct $q$ reaches a round at least $r$ by $t_1 + \Delta \leq t_p + \Delta$ and holds the proposal by $t_p + \Delta$, before $\mathit{OnTimeoutPropose}(r)$ fires at $q$, which happens no earlier than $t_1 + \tprop(r) > t_1 + 2\Delta \geq t_p + \Delta$.

Up to $t_p + \Delta$, every correct precommit in round $r$ is for $\id(x)$: consider the first correct precommit in round $r$ that is not; it is not a timeout; it is not an exit precommit, since a quorum of round $r$ holds at least $n-2f \geq 2f+1$ correct precommits, all for $\id(x)$ before this event, and no other value has $2f+1$ in it; and it is not a rejection of $p$'s proposal, because $J$ verifies, $\valid(x)$ holds when $vr = -1$ since $p$ is correct, and $\id(x)=x_J$ when $vr\geq0$.

Hence every correct process precommits $\id(x)$ in round $r$ by $t_p + \Delta$, on the proposal or as an exit precommit on a quorum of round $r$.

These at least $n-f$ precommits reach every correct process by $t_p+2\Delta$. A fresh $x$ travels in the proposal. For a reproposal, the support held by $p$ at $t_p$ contains an earlier correct precommit for $\id(x)$. The earliest correct such precommit followed an application-valid fresh proposal held by a correct process before $t_p$, so gossip supplies the full value by $t_p+\Delta$. Every correct process decides by $t_p + 2\Delta \leq t_1 + 3\Delta$.

Replacing $\Delta$ by $\delta$ in catch-up, in the delivery of the proposal and in the delivery of the precommits gives the bounds with $\delta$; the timeout condition is still $\tprop(r) > 2\Delta \geq 2\delta$. Termination follows as in Theorem~\ref{thm:termination}, with Lemma~\ref{lem:good} replaced by the bound just proved.
\end{proof}

Between two consecutive rounds with correct proposers after $\GST$, Justify waits for messages only. A round whose proposer sends nothing costs $\tprop$, as a Byzantine proposer costs a round timer in any protocol of this kind, and one message delay for the quorum of nil precommits.

The justification carried by the next proposal is one certificate plus one dead certificate per dead round since it; a mixed dead certificate can be larger than a support certificate. With aggregate signatures a certificate or a nil certificate is one signature and a bitmap, and a mixed dead certificate is one aggregate per distinct value it contains.

A correct proposer's proposal carries its exit quorums for $vr<j<r$. A lagging process already in one of those rounds can reach $r$ from that evidence; a process at or below $vr$ still needs the missing round quorums. A general valid justification may contain smaller dead certificates and need not suffice for catch-up.

\section{Advancing on \texorpdfstring{$2f+1$}{2f+1} precommits}
\label{app:fastexit}

Walk and Justify leave a round on a quorum of $n-f$ precommits of that round. This appendix replaces the quorum by a support certificate or a dead certificate, separating round advancement from the decision threshold as in Minimmit~\cite{minimmit}, and shows that the proofs of Appendices~\ref{app:walk} and~\ref{app:justify} carry over. The change applies to Walk and to Justify; timed Jump in Appendix~\ref{sec:jumpcarry} uses this version of Walk and additionally advances on a timed support certificate of a later round. We do not analyze the change for Patch or Carry. The proof below identifies the required changes to the proofs in Appendices~\ref{app:walk} and~\ref{app:justify}.

\paragraph{The rule.} Replace Rule~\ref{ln:exit} of Algorithm~\ref{alg:walk}, and Rule~\ref{ln:jexit} of Algorithm~\ref{alg:justify}, by the two rules below. Rule~\ref{ln:obs} of Walk stays and is applied before them.
\begin{itemize}[nosep]
\item \emph{Certificate exit.} Upon a certificate for some $x$ in round $\round_p$, for the first time in this round: if $\step_p = \propose$ then precommit $\id(x)$ and set $\step_p \gets \pcstep$; if $\decision_p = \nil$ then $\mathit{StartRound}(\round_p+1)$.
\item \emph{Dead exit.} Upon a dead certificate for round $\round_p$: if $\decision_p = \nil$ then start round $\round_p+1$.
\end{itemize}
The certificate exit has priority over the dead exit. In Justify, $D$ then holds, for each round $j$ in $(vr, r)$, the dead certificate on which the proposer left round $j$.

A round whose proposer sent nothing now ends $2f+1$ nil precommits after the proposal timeout rather than $n-f$; a round whose proposal $2f+1$ processes accepted ends on their precommits, and a decision still requires $n-f$ matching precommits through the unchanged decision rule; the exit precommit is never $\nil$, since a dead exit sends nothing. No second message type and no second signature per round is introduced.

The threshold $2f+1$ suffices for matching non-nil precommits or for nil precommits; it does not suffice for an arbitrary mixed set. A mixed dead certificate can be larger.

Minimmit instead uses either $2f+1$ matching block votes, an M-notarization, or $2f+1$ dedicated nullify messages~\cite[Section~3 and Algorithm~1]{minimmit}. A process may nullify after voting only on messages from $2f+1$ distinct senders, each a nullification or a vote opposing its vote; a timeout permits nullification only before voting, and voting after nullification is forbidden. The extension here uses ordinary precommits and does not adopt that additional signed-message mechanism.

\begin{theorem}
\label{thm:fastexit}
Let $n \geq 5f+1$ and at most $f$ Byzantine processes. Walk and Justify with the certificate exit and the dead exit satisfy Agreement, Validity and Termination, under the assumptions of Theorems~\ref{thm:termination} and~\ref{thm:responsive} respectively. Their Agreement holds already for $n \geq 3f+1$.
\end{theorem}

\begin{proof}
We list the places where the proofs of Appendices~\ref{app:walk} and~\ref{app:justify} use the quorum, and what replaces it. Fix $v$ with $n-f$ precommits in round $r$. By Lemma~\ref{lem:evidence}(iv), the only certificate of round $r$ is for $v$ and no dead certificate for round $r$ exists; this uses only $n-f \geq 2f+1$, which is where the bound $n \geq 3f+1$ for Agreement comes from.

\emph{Agreement of Walk.} In the proof of Theorem~\ref{thm:walk-safe}, the third property of Algorithm~\ref{alg:walk} becomes: a process reaches a round above $r$ only by leaving round $r$ through the certificate exit or the dead exit. In the (I1) case, the dead exit is impossible in round $r$, and the certificate exit needs a certificate for some $x$ in round $r$, so $x = v$ and Rule~\ref{ln:obs} has recorded $(r, \id(v))$ as before. In the (I2) case, an exit precommit for $\id(x)$ in round $s \geq r$ needs a certificate for $x$ in round $s$, contradicting (I3). Nothing else changes, and the counting that every quorum of round $r$ contains a certificate for $v$ is no longer needed.

\emph{Validity.} An exit precommit still needs $2f+1$ precommits for the same identifier in the same round, hence $f+1$ earlier correct ones.

\emph{Agreement of Justify.} The proof of Theorem~\ref{thm:justify-safe} uses the quorum only in the exit precommit case, which is replaced as above.

\emph{Termination.} Lemma~\ref{lem:catchup} holds with $E_j$ the certificate or dead certificate on which $p$ left round $j$. Lemma~\ref{lem:ends}: once every correct process has precommitted in round $s$, their at least $n-f$ precommits contain a certificate or form a dead certificate by Lemma~\ref{lem:exit}, which reaches every correct process; and a process that leaves $s$ on a dead certificate before precommitting holds evidence that reaches every correct process as well.

Lemma~\ref{lem:proposer} retains its two alternatives. An intermediate state between observation and the certificate exit can still have $\validRound_q=r$. In the good-round first-conflict argument, that observation uses a certificate containing an earlier correct precommit for the proposed value, so acceptance follows by value equality as in the main proof.

In Lemma~\ref{lem:good}, before the broadcast only Byzantine precommits exist in round $r$, so neither a certificate nor a dead certificate for round $r$ exists and the proposer's wait is not aborted; up to $H$, a conflicting exit precommit needs a certificate for its value in round $r$ and hence $f+1$ earlier correct precommits for it, a dead certificate for round $r$ needs $2f+1$ senders of something other than $\id(x)$ and only the Byzantine processes qualify, and an exit on a certificate is an exit on the certificate for $x$ with an exit precommit for $\id(x)$; the rest is unchanged. Theorem~\ref{thm:responsive} changes in the same two places.
\end{proof}

\section{Alternative round-change designs}
\label{app:alternatives}
\label{sec:five}

Walk and Justify are the two constructions developed and proved in this paper. This appendix retains three other designs only as sketches, to explain the tradeoffs around round skipping. We make no complete Agreement, Validity or Termination claim for these sketches. The optional all-nil shortcut of Appendix~\ref{sec:shortcut} remains disabled.

\begin{table}[ht]
\centering
\small
\begin{tabular}{@{}p{0.13\linewidth}p{0.38\linewidth}p{0.42\linewidth}@{}}
\toprule
Design & Round-change mechanism & Implication \\
\midrule
Patch & Observe support in current and future rounds; retain quorum jumps. & Some processes still skip without the protected value; bound how many can do so. \\
Timed Jump & Import a future support certificate, then jump after a timer. & Needs current-round fast exits, a broadcast guard and a participation argument. \\
Carry & Attach the recorded certificate to precommits; import the highest certificate before a quorum jump. & Larger votes and verification of carried evidence. \\
\bottomrule
\end{tabular}
\caption{Alternative design sketches. All retain proposal timeouts and a timer for round advancement. Patch and Carry retain the published $\tprec$ wait and use it for $\mathit{WaitForValid}$; timed Jump uses $\twait$ at the proposer and $\tprec$ for future-certificate jumps. None is needed by Walk or Justify.}
\label{tab:cmp}
\end{table}

\subsection{Patch}

\label{sec:patch}

Replace the observation rule by: on $2f+1$ precommits for $\id(x)$ in a round $r \geq \round_p$ with $r > \validRound_p$, set $\validPair{p} \gets (r,\id(x))$. Nothing else changes. The rule must not count rounds below $\round_p$: a late-completed certificate for round $r-1$, delivered after the proposer of round $r$ has left $\mathit{WaitForValid}$, would be adopted by every other correct process, which would then reject the correct proposer's value, and Lemma~6 of~\cite{ft} would fail.

The safety intuition starts with the same count as Walk: every quorum of a decided round contains a certificate for its decided value. The patched observation rule records that certificate even at a process still in an earlier round. Unlike Walk, some correct processes may still skip the decided round without recording it. The first quorum formed above that round must be in its immediate successor; its correct signers already recorded the decided value, leaving at most $f$ correct processes outside it. Together with the Byzantine processes, those outsiders cannot supply $2f+1$ conflicting precommits. Turning this intuition into a proof requires simultaneous invariants for recorded state and later certificates.

\subsection{Timed Jump}

\label{sec:jumpcarry}

Start from Walk with the certificate and dead exits of Appendix~\ref{app:fastexit}, rather than the quorum exit of Algorithm~\ref{alg:walk}. Add the rule: on a certificate for $x$ in a round $s > \round_p$, for the first time for round $s$, set $\validPair{p} \gets (s,\id(x))$ if $s > \validRound_p$ and schedule $\mathit{OnTimeoutJump}(s)$ after $\tprec(s)$. The handler calls $\mathit{StartRound}(s+1)$ only if $s \geq \round_p$ and $\decision_p = \nil$. If $p$ is the current proposer and has not broadcast, receipt of this future certificate also aborts its proposer continuation. The current-round certificate and dead exits stay, including the exit precommit on a certificate and the priority of that exit. They also abort a pending proposer continuation.

Immediately before a proposal broadcast from $\mathit{StartRound}(r)$, check that $\round_p=r$, $\step_p=\propose$, $\decision_p=\nil$, and $\validRound_p<r$. If any check fails, return without broadcasting. This check, selection of the proposal and broadcast are one atomic action. The check is unconditional: a process can already hold a future certificate when it enters a round in which it is the proposer. Such a process waits for a current-round exit or its pending jump timer; it must not broadcast a proposal with $vr\geq r$. A round change also cancels the old proposer continuation.

A process still precommits only in its current round. During a pending jump $\validRound_p$ may exceed $\round_p$, so Walk's invariant $\validRound_p\leq\round_p$ cannot be reused. The safety intuition is that the imported certificate preserves the protection of any earlier decision. The timer addresses a separate issue: a process must have time to participate in a good round before skipping it. A termination argument needs the Walk timeout bounds and, for every sufficiently large $s$, finite nondecreasing $\tprec(s)>2\Delta$. Since no future certificate exists before the good proposer broadcasts, that wait prevents a jump before the proposal arrives. These are proof obligations for this sketch, not a termination theorem here. A lagging process can cross a certified stretch using one certificate and one wait.

The fast exits of Appendix~\ref{app:fastexit} are essential to this design: importing a future certificate can move a process ahead while other processes lack the $n-f$ precommits needed for a quorum exit. A support certificate must also suffice to release those processes from its own round. Inheriting the quorum-only exits of Algorithm~\ref{alg:walk} is therefore not an interchangeable choice.

\subsection{Carry}

Every precommit $\langle\PRECOMMIT, s, x, C\rangle$ carries $C$, the certificate behind the sender's valid state, or $\bot$; receivers verify $C$ and discard precommits with invalid certificates. Verified enclosed precommits and certificates count as held evidence for $\mathit{WaitForValid}$, the current-round observation rule and the round-change import below. This makes attachments available to the existing rules; it does not add an unconditional state update on receipt. Every count is by distinct sender in the same height and round.

Keep the published quorum rule and the timer, and when $\mathit{OnTimeoutPrecommit}(s)$ fires on a quorum $Q$ of round $s \geq \round_p$, let $C^\ast$ be the highest-round certificate among the certificate contained in $Q$ if some value has $2f+1$ precommits in $Q$, and the certificates carried by the precommits of $Q$; if $C^\ast$ is for a round above $\validRound_p$, record its round and value identifier in $\validPair{p}$ before starting round $s+1$.

Vote size grows by a certificate, or by an aggregate signature; this is the highQC of HotStuff~\cite{hotstuff} placed in the vote.

The safety intuition is inductive. A quorum of the decided round contains support for the decided value directly. In a later-round quorum, a correct signer that retained the decision's protection carries a certificate with that protection, allowing the receiver to import it before skipping. A complete proof must establish that property for the signers before using it for the receiver. The published post-quorum timer remains important for participation; deleting it is not justified by the import rule alone.

For Patch and Carry, a termination argument must retain the published post-quorum timer and the proposer's evidence wait, with finite, positive, nondecreasing timeouts eventually satisfying $\tprec(r)>2\Delta$ and $\tprop(r)>\tprec(r)+2\Delta$. Verified attachments must count as held evidence. These conditions identify the intended setting for the sketches; the formal termination results in this paper are those for Walk, Justify and their optional fast exits.

\section{Deleting the precommit wait}
\label{sec:timer}

In Tendermint and in~\cite{ft}, a process that holds a quorum of its round waits $\tprec$ before moving on. The wait gives a delayed proposal time to arrive and receive a precommit. Walk has no such post-quorum wait; it still uses $\tprop$ and $\twait$. Justify still uses $\tprop$. This section shows that the wait cannot simply be deleted from the published protocol, from Patch, from Jump or from Carry, and explains how Walk avoids these executions; Walk's termination is proved in Appendix~\ref{app:walk}.

\begin{theorem}
\label{thm:timer}
Let $n \geq 5f+1$ and $f \geq 1$. Remove the precommit timer from the published protocol, from Patch and from Carry, so that a process starts round $s+1$ as soon as it holds a quorum of a round $s \geq \round_p$; and from Jump, so that a process starts round $s+1$ as soon as it holds a certificate of a round $s > \round_p$, either without precommitting or precommitting for the certified value in round $s$ as it does so. Each of these variants has an execution after $\GST$ in which no correct process ever decides.
\end{theorem}

\begin{proof}
Let $\GST=0$, let exactly $f$ processes be Byzantine, and let $\tau_0>0$ be a lower bound on the proposal timeouts. Choose $\varepsilon>0$ with $8\varepsilon<\min\{\Delta,\tau_0\}$. The schedules below have positive microsteps; their timing is specified at the end of the proof.

In the published protocol, Patch and both forms of Jump, every copy of withheld evidence obeys the stated delay, including copies in forwarded messages and proposal attachments.

Carry needs one adjustment: a later precommit itself carries an older support certificate. Permit these attached certificates to arrive, choosing their signers recursively from the correct, non-nil voters, while withholding the old Byzantine nil precommits needed to complete the earlier $n-f$ quorum. The attached support neither makes Carry vote in the old round nor causes it to leave that round.

Let $\proposer(0)$ be correct. The $f$ Byzantine processes precommit $\nil$ in every round at once and, whenever one of them is the proposer, repropose the value certified in the previous round at once, with that certificate.

A decision needs $n-f$ precommits for one value in one round, hence a precommit from every correct process in that round. We give schedules in which every round lacks the precommit of one correct process, the \emph{absentee} of the round. Each intended voter receives its proposal before its proposal timeout; each absentee skips its missing round before a timeout. No correct process precommits $\nil$.

\emph{Schedule A: exits that send no precommit} (the published protocol, Patch, Carry, and Jump without a precommit on jumping). Choose a correct $a_0 \notin \{\proposer(0), \proposer(1), \proposer(2)\}$, let $a_1 := a_0$, and for every $j \geq 2$ choose a correct $a_j \notin \{\proposer(j), \proposer(j+1), a_{j-1}\}$; at least $4f+1 \geq 5$ correct processes exist, so this is possible.

Deliver every message at once, with these exceptions: to $a_0$, the proposal of round $0$, the precommits of round $0$ and the proposal of round $1$ arrive only after a quorum of round $1$ has reached $a_0$; to $a_j$ with $j \geq 2$, the precommits of round $j-1$ from other processes and the proposal of round $j$ arrive only after a quorum of round $j$ has reached $a_j$, and the precommits of round $j+1$ arrive only after $a_j$ has precommitted in round $j+1$.

The proposals are delayed because a reproposal may carry the certificate of the previous round, which is exactly the evidence being withheld. For Carry, allow the attached support as specified above and withhold the old Byzantine nil votes instead. On its initial return, $a_0$ receives the round-2 proposal and precommits before receiving round-2 exit evidence. All these delays stay below $\Delta$.

Round~0: $\proposer(0)$ proposes a value $v$. Every correct process other than $a_0$ precommits $\id(v)$ at once; with the Byzantine $\nil$ precommits this is a quorum of round $0$, which contains a certificate for $v$. It reaches every process but $a_0$ at once, and they enter round $1$ holding $\validPair{} = (0,\id(v))$, by the observation rule in the published protocol, Patch and Carry and by Rule~\ref{ln:obs} in Jump. $\proposer(1)$ is among them and reproposes $v$ at once, since $\validRound = 0$ makes $\mathit{WaitForValid}$ return immediately; every correct process other than $a_0$ precommits $\id(v)$ in round $1$ at once.

The resulting quorum of round $1$ reaches $a_0$, still in round $0$ and in step $\propose$, before anything else, and $a_0$ moves to round $2$ without having precommitted in rounds $0$ and $1$. Rounds $0$ and $1$ end with $n-f-1$ precommits for $v$ each.

Round~$j \geq 2$: we maintain that every correct process other than $a_j$ is in round $j$, that $a_{j-1}$ is among them, and that $a_j$ is in round $j-1$ in step $\pcstep$, having precommitted there and still waiting for the round-$(j-1)$ precommits of the others. This holds for $j = 2$ after round $1$, with $a_0$ in the role of $a_{j-1}$.

$\proposer(j)$ is not $a_j$ and holds $\validPair{} = (j-1,\id(v))$, having left round $j-1$ on a quorum that contained the certificate for $v$; it reproposes $v$ at once. Every correct process other than $a_j$ precommits $\id(v)$ at once; $a_{j-1}$ does so because it is already in round $j$ when the proposal arrives, having moved there when the quorum of round $j-1$ formed. The quorum of round $j$, which contains the certificate for $v$, reaches $a_j$ at once, before the round-$(j-1)$ precommits, and $a_j$ moves to round $j+1$ without having precommitted in round $j$; every correct process other than $a_{j+1}$ enters round $j+1$ as well, while $a_{j+1}$ remains in round $j$.

$\proposer(j+1)$ is not $a_j$ and reproposes $v$ at once, every correct process other than $a_{j+1}$ precommits $\id(v)$ in round $j+1$, $a_j$ among them, and $a_{j+1}$, which is not $a_j$, has precommitted in round $j$ and now waits for the round-$j$ precommits of the others. Thus the invariant holds for $j+1$. Round $j$ ends with $n-f-1$ precommits for $v$.

\emph{Schedule B: Jump with a precommit in the certificate's round.} Here a process that jumps on a certificate of round $s$ precommits in round $s$ but not in the rounds it passes below $s$, so the absentee of round $j$ must be a process that is in round $j$ and jumps on a certificate of round $j+1$. Choose correct $q_j\notin\{\proposer(j),\proposer(j+1),\proposer(j+2),q_{j-1}\}$ for every $j\geq0$, omitting $q_{-1}$ when $j=0$. There are at most four exclusions and at least five correct processes.

Deliver every message at once, except that the proposal of round $j$, the precommits of round $j$ from other processes and the proposal of round $j+1$, which may carry the certificate of round $j$, reach $q_j$ only after a certificate of round $j+1$ has reached it.

In round $j$, the correct processes other than $q_j$ and $q_{j-1}$, at least $n-f-2 \geq 2f+1$ of them, precommit $\id(v)$ at once and form a certificate, on which they leave round $j$; $q_{j-1}$ is in round $j-1$ and, on that certificate, jumps to round $j+1$ precommitting $\id(v)$ in round $j$.

$\proposer(j+1)$, which holds $\validPair{} = (j,\id(v))$, reproposes $v$ at once; the correct processes other than $q_{j+1}$ and $q_j$ precommit it at once and form a certificate of round $j+1$, which reaches $q_j$ before any message of round $j$, so $q_j$ jumps to round $j+2$, precommitting in round $j+1$ and never in round $j$. Round $j$ therefore holds precommits for $v$ from every correct process except $q_j$: $n-f-1$. For $j = 0$ the proposer proposes fresh and $q_{-1}$ does not exist; the counts only improve.

The schedules can be realized with the following timing. Let the proposer broadcast in round $j$ at $t_j=4j\varepsilon$, deliver its proposal to the intended voters at $t_j+\varepsilon$, and deliver their exit evidence to the permitted recipients at $t_j+2\varepsilon$. Delay that evidence to $\proposer(j+1)$ until $t_j+4\varepsilon$, when it enters and proposes in the next round.

The absentee choices ensure that this proposer is not a process whose jump must happen earlier; in Schedule B this is why $q_j$ also excludes $\proposer(j+2)$. Release the withheld older evidence to an absentee one microstep after its prescribed skip. Initially $a_0$ receives the round-1 quorum at $6\varepsilon$ and the withheld round-0 evidence and round-1 proposal at $7\varepsilon$.

Each withheld message, including each forwarded copy, is delivered within $8\varepsilon$ of its first correct receipt. All proposal timeouts are longer than this delay, and time advances by $4\varepsilon$ per round.

Every round has exactly $n-f-1$ non-nil precommits, all for $v$, and Byzantine processes send only $\nil$. Hence no decision certificate exists, although every proposer sends a valid, non-equivocating proposal and all message delays satisfy the gossip bound. This proves that deleting the specified timers permits nontermination. This negative result concerns only the specified timer-deleted rules; it does not establish correctness of the timed sketches in Appendix~\ref{app:alternatives}.
\end{proof}

Neither schedule prevents Walk from deciding. A process leaves a round only on a quorum of that round, and in the executions above every quorum of round $j$ contains the certificate for $v$ and no $2f+1$ precommits for anything else, so the exit precommit is for $v$ and every correct process precommits $\id(v)$ in every round. We do not claim that every rule that skips rounds breaks termination: a rule that skips only on a quorum of a later round, which is also a decision, is harmless, since the skipping process decides.

Consider three ways to ensure that a process precommits in every round it passes:
\begin{enumerate}[nosep]
\item Send a late precommit for a skipped round once its certificate arrives. This requires signers to accept a round below the last one signed.
\item When jumping, precommit the certified value in every skipped round. This costs one signature per skipped round. The intended safety argument is that the certified value equals the decided value in every skipped round above the decision round; we do not develop that argument here.
\item Advance without skipping, and cast the exit precommit on the round's own evidence.
\end{enumerate}
The schedules of Theorem~\ref{thm:timer} fail against all three, since the absentee then precommits in the round it was meant to miss. We prove termination only for the third, which Walk uses.

\section{The optional all-nil shortcut}
\label{sec:shortcut}

After a round whose proposer sent nothing, every correct process precommits $\nil$ on its proposal timeout or before leaving on a quorum; the quorum can include Byzantine precommits for values. The next proposer still waits in $\mathit{WaitForValid}$, up to $\twait$, for a certificate of that round. Lemma~\ref{lem:evidence}(iii) licenses a shortcut: $\mathit{WaitForValid}(r)$ may return as soon as $p$ holds a quorum of round $r-1$ consisting of $\nil$ precommits only, since then no process can hold a certificate of round $r-1$. 

The shortcut preserves Agreement, since it only shortens a wait. It is not part of Algorithm~\ref{alg:walk} or any termination theorem in this paper. 

The shortcut can delay a decision by several rounds. A process that is still catching up may record a certificate of an earlier round less than $\Delta$ before the shortcut fires and reject the fresh proposal; and when correct processes hold certificates of different old rounds, several correct proposers in a row can take the shortcut and repropose their own, each rejected by the holders of a higher one, as long as the rounds end on all-nil quorums faster than the certificates propagate. Termination with the shortcut would need its own argument, which we do not give; leave it disabled in the implementation specified here.

An all-nil quorum is not the same as all processes voting nil: with $n=6$ and $f=1$, one correct process may vote for a value while the other four correct processes and the Byzantine process supply five nil precommits. A quorum of round $r-1$ that contains precommits for a value does not license the shortcut, however few: a faulty proposer that sends its value to two correct processes and precommits it as well can complete a certificate for it even though $p$'s quorum shows two precommits for it, and $p$ must wait.

The shortcut can avoid the full wait when the proposer holds an all-nil quorum; otherwise the usual return conditions remain, including immediate return on a certificate for round $r-1$. Section~\ref{sec:justify} removes the wait altogether.

\section{Extension to chains}
\label{sec:chains}

This appendix proves prefix Agreement for the chained variant of Walk and argues Validity; it sketches the same construction for Justify. It does not claim continuing progress, whose requirements are listed at the end.

In the per-height protocol, a proposer that selects a certified value reproposes it rather than extending it. Under Walk, every correct process that passes round $r$ holds a support certificate for any value decided in it, so the proposer may extend the certified value instead of repeating it.

Let values be blocks in a tree with a genesis block, write $x \succeq y$ if $y$ is $x$ or an ancestor of $x$, and let deciding $x$ also commit every ancestor of $x$. Rounds are no longer grouped into heights, and $\validRound_p$ and $\validValueId_p$ are never reset. When $\validValueId_p\neq\nil$, write $b_p$ for the block it identifies; this notation does not add protocol state. A block authenticates its parent identifier and creation round. Correct proposers create a strict new child of the selected parent, and parent links and any ancestry used in an extension check are verified.

Change Algorithm~\ref{alg:walk} as follows. Replace $\decision_p$ by $\committedTip_p$, the committed tip, initially genesis; remove the per-height decision guards and stopping condition, and do not cancel a proposer continuation on commitment; Rule~\ref{ln:exit} starts the next round unconditionally, since rounds continue past decisions.

In $\mathit{StartRound}$, the proposer proposes a new block $x$ whose parent is $b_p$ if $\validRound_p\geq0$ and genesis otherwise, and broadcasts $\langle\PROPOSAL, r, x, \validRound_p\rangle$; a proposal with $vr \geq 0$ comes with the certificate for the parent $y$ in round $vr$, and $x \succeq y$.

Rule~\ref{ln:fresh} accepts $x$ if $\valid(x)$ and either $\validRound_p = -1$ or $x\succeq b_p$. Rule~\ref{ln:repro} accepts $x$ with a certificate for $y$ in round $vr$ if $x\succeq y$, $\valid$ holds for the part of $x$ beyond $y$, and either $\validRound_p \leq vr$ or $x\succeq b_p$; it sets $\validPair{p} \gets (vr, \id(y))$ when $\validRound_p \leq vr$.

Rule~\ref{ln:decide} becomes: on $n-f$ precommits for $\id(x)$ in some round, with $x$ and its uncommitted ancestry available and their identifiers and parent links verified, if $x \succeq \committedTip_p$ then set $\committedTip_p \gets x$ and commit the ancestors of $x$ not yet committed.

Support certificates from rounds below $\validRound_p$ cannot change the recorded valid state: Rule~\ref{ln:obs} sees the current round, line~\ref{ln:wfv} sees rounds above $\validRound_p$, and Rule~\ref{ln:repro} accepts a parent certified below $\validRound_p$ only when the new block extends $b_p$. An older decision certificate can still advance $\committedTip_p$ if its block extends the current tip.

\begin{theorem}
\label{thm:chain}
Let $n \geq 5f+1$ and at most $f$ Byzantine processes. In chained Walk, if $v$ obtains $n-f$ precommits from distinct processes in round $r$, then every block with a certificate in a round $s \geq r$ extends $v$. Consequently any two decided blocks are comparable under $\succeq$: decided blocks form one chain.
\end{theorem}

\begin{proof}
Replace, in the statements (I1) to (I3) of the proof of Theorem~\ref{thm:walk-safe}, ``$\validValueId_p=\id(v)$'' by ``$b_p\succeq v$'', ``$x \neq v$'' by ``$x \not\succeq v$'', drop the clause about decided processes in (I1), since processes keep running, and keep the counting for round $r$, where the only certified block is $v$ itself. The case analysis is the same with one more argument per rule; the decision rule changes no state that the invariants mention.

Rule~\ref{ln:fresh} at a process above $r$ needs $x\succeq b_p\succeq v$, so $x \succeq v$. Rule~\ref{ln:repro} at a process above $r$ needs either $x\succeq b_p\succeq v$, or $\validRound_p \leq vr$, which gives $vr \geq r$, a certified parent $y \succeq v$ by (I3), and $x \succeq y \succeq v$; its write to $\validPair{p}$ is $(vr, \id(y))$ with $y \succeq v$.

A non-nil exit precommit (line~\ref{ln:exitvote}) is for a block with $2f+1$ precommits in a quorum of round $\round_p \geq r$, hence certified there, and Rule~\ref{ln:obs} records a certificate of round $\round_p \geq r$; both blocks extend $v$ by (I3). Line~\ref{ln:wfv} adopts a certificate of a round above $\validRound_p \geq r$, likewise. Transitivity of $\succeq$ closes each case.

For the consequence, let $v$ and $w$ be decided in rounds $r \leq s$; $w$ has a certificate in round $s$, so $w \succeq v$.
\end{proof}

\paragraph{Validity.}
For Validity, induct on the first correct precommit for each block in the causal order of the execution. A non-nil exit precommit requires earlier support for the same block and therefore an earlier correct precommit. Otherwise a correct process accepts a proposal. A fresh proposal is application-valid by the fresh-proposal check. A proposal extending a certified parent has an application-valid parent by induction, and its new extension is checked before the correct precommit. Parent links are authenticated by the block identifiers. Every decision certificate includes a correct precommit, so committed blocks and their ancestors are valid.

\paragraph{Continuing progress.}
Continuing progress for this extension requires a separate argument. The required statement is that every correct process's committed-chain height grows without bound. Parent and ancestor data must be retrievable, and the time needed to construct and validate an application-valid child must be included in the timeout bounds. Correct proposers must create strict new children, and the application must permit a valid child, including an empty block when appropriate. Transaction inclusion additionally requires a fair proposal policy and assumptions about continued transaction validity. These requirements are not established by the per-height Termination proofs, so this appendix claims prefix Agreement and Validity, but not a complete SMR progress theorem.

\paragraph{Justify.}
The same construction applies to Justify: a proposal carries a new full block $x$, the support certificate of its parent and the dead certificates for intervening rounds. Replace the certified-identifier equality test by checking that $x$ is an application-valid child of that certified parent; with $vr=-1$, the parent is genesis.

Replace equality with extension in the proof of Theorem~\ref{thm:justify-safe}; a conflicting block's justification either needs a nonexistent dead certificate for a decided round or earlier support for a conflicting block. The first such conflict is therefore impossible. The validity induction above applies; constructing the child is subject to the data and execution requirements above, so the immediate-broadcast and latency claims of Theorem~\ref{thm:responsive} do not automatically carry over.

When every round certifies a block, the justification contains only the parent certificate. A certified but undecided block can therefore be extended rather than reproposed, while committed blocks form a common prefix.

\section{Implementation assumptions and related protocols}
\label{app:implementation}

\paragraph{Realizing the gossip property.}
The timing arguments use gossip to deliver round-change evidence to lagging processes, recorded certificates to a proposer, and the value behind a reproposed identifier to deciding processes. Over partially synchronous links, this requires retaining and retransmitting signed evidence and proposed values, including Byzantine-origin precommits and certificates carried inside other messages. Information received before $\GST$ must still be recovered after $\GST$.

One-time forwarding or an unspecified fetch-by-identifier operation does not establish the bound. Retrieval, retransmission and relay delays must be included in $\Delta$, and evidence and data must remain available after a local decision. Reproposals may attach their support certificate; Justify attaches the complete justification by construction. These are implementation requirements for the stated dissemination abstraction, not additional safety assumptions.

\paragraph{Signers and implementation.} Under Walk and Justify a process precommits at most once per round and only in its current round, so signing rules that enforce monotone (height, round, step) are unchanged. Vote accounting must track support per (round, value), since two values can each have $2f+1$ precommits in a round, and must admit an equivocator's conflicting precommit when a peer reports support for the other value. Precommits of rounds ahead are cached.

A recovery implementation must durably log or reconstruct sent precommits, valid-state updates, round changes and their evidence before resuming consensus. Logging the valid pair only when signing is insufficient: a later observation can update it before the next signature. Justify has no valid pair, but still needs its sent-vote and round-change evidence. The proofs assume correct processes retain state; crash recovery must restore the same abstract state.

Processes sending an exit precommit may lack the block data of the certified value and fetch it before deciding; Validity does not depend on them having it.

\paragraph{Adding a lock.} Tendermint sets its lock after observing matching prevotes and before sending a non-nil precommit. Restoring that rule would require the prevote step that Fast Tendermint removes. An alternative would be to lock whenever a process sends a non-nil precommit. Such a lock would retain information at the processes in the second case of Section~\ref{sec:violation}, but not at processes that skip the round without precommitting. It would also require a safe unlocking rule. We do not specify that alternative here.

\paragraph{Related protocols.} Minimmit~\cite{minimmit} advances views on a notarization, whose threshold is below the finalization threshold, or on a nullification assembled from dedicated nullify messages; its notarization and finalization thresholds correspond to our support and decision certificates, while the certificate contents and voting rules differ. Our nil and dead certificates are made of ordinary precommits, and a correct process sends only one precommit per round. Timed Jump uses a support certificate of a later round to schedule round advancement and never precommits in an earlier round.

Simplex~\cite{simplex} enters an iteration on a notarized chain; the justification of Section~\ref{sec:justify} plays that role. HotStuff~\cite{hotstuff} carries the highest quorum certificate in votes and proposals; Carry and Justify are the two placements.

Kudzu~\cite{kudzu} and Hydrangea~\cite{hydrangea} keep a slow path at $f < n/3$; BlueBottle~\cite{bluebottle} decides over a DAG. The bound $n \geq 5f-1$ for good-case decisions in two communication steps~\cite{goodcase,kuznetsov} is not approached here.

Model checking Walk and Justify against the invariants (I1) to (I3) remains future work.


\begin{thebibliography}{99}
\small
\setlength{\itemsep}{2pt}
\setlength{\parsep}{0pt}
\bibitem{tendermint} E.~Buchman, J.~Kwon, Z.~Milosevic. The latest gossip on BFT consensus. arXiv:1807.04938v3, 2019. Originally submitted in 2018.
\bibitem{ft} P.~Vander Vos, D.~Cason. Fast Tendermint: speeding up a foundational consensus protocol. arXiv:2608.13434v1, 2026.
\bibitem{minimmit} B.~Kobayashi Chou, A.~Lewis-Pye, P.~O'Grady. Minimmit: fast finality with even faster blocks. arXiv:2508.10862v7, 2026. Originally submitted in 2025.
\bibitem{simplex} B.~Y.~Chan, R.~Pass. Simplex consensus: a simple and fast consensus protocol. TCC 2023. Cryptology ePrint Archive 2023/463.
\bibitem{hotstuff} M.~Yin, D.~Malkhi, M.~K.~Reiter, G.~Golan Gueta, I.~Abraham. HotStuff: BFT consensus with linearity and responsiveness. PODC 2019.
\bibitem{kudzu} V.~Shoup, J.~Sliwinski, Y.~Vonlanthen. Kudzu: fast and simple high-throughput BFT. arXiv:2505.08771, 2025.
\bibitem{hydrangea} N.~Shrestha, A.~Kate, K.~Nayak. Hydrangea: optimistic two-round partial synchrony with improved fault resilience. Cryptology ePrint Archive 2025/1112.
\bibitem{bluebottle} P.~Vander Vos, A.~Sonnino, G.~Tsimos, P.~Jovanovic, L.~Kokoris-Kogias. BlueBottle: fast and robust blockchains through subsystem specialization. arXiv:2511.15361, 2025.
\bibitem{goodcase} I.~Abraham, K.~Nayak, L.~Ren, Z.~Xiang. Good-case latency of Byzantine broadcast: a complete categorization. PODC 2021. doi:10.1145/3465084.3467899.
\bibitem{kuznetsov} P.~Kuznetsov, A.~Tonkikh, Y.~X.~Zhang. Revisiting optimal resilience of fast Byzantine consensus. PODC 2021. doi:10.1145/3465084.3467924.
\end{thebibliography}
\end{document}